\documentclass[10pt,reqno]{amsart}

\usepackage{amsmath}
\usepackage{amssymb}
\usepackage{amsfonts}
\usepackage{epsfig}
\usepackage[margin=1.15in]{geometry}

\newtheorem{theorem}{Theorem}[section]
\newtheorem{lemma}{Lemma}[section]

\newtheorem{remark}{Remark}[section]

\newtheorem{example}{Example}[section]

\numberwithin{equation}{section}

\begin{document}
\title{Asymptotic Investment Behaviors under a Jump-Diffusion Risk Process}
\author{Tatiana Belkina}
\thanks{Laboratory of Risk Theory, Central Economics and Mathematics Institute  of  the Russian Academy of Sciences and
National Research University Higher School of Economics, Moscow, Russia (TB)}
\author{Shangzhen Luo}
\thanks{Department of Mathematics, University of Northern Iowa, Cedar Falls, IA
50614 USA (SL)}
\address{Correspondence Address:
Department of Mathematics, University of Northern Iowa, Cedar Falls, IA
50614-0506 USA}
\email{luos@uni.edu}

\begin{abstract}
We study an optimal investment control problem for an insurance company.
The surplus process follows the Cramer-Lundberg process
with perturbation of a Brownian motion.
The company can invest its surplus into a risk free asset and a Black-Scholes risky asset.
The optimization objective is to minimize the probability of ruin.
We show by new operators that the minimal ruin probability function
is a classical solution to the corresponding HJB equation.
Asymptotic behaviors of the optimal investment control policy
and the minimal ruin probability function are studied 
for low surplus levels with a general claim size distribution.
Some new asymptotic results for large surplus levels
in the case with exponential claim distributions are obtained.
We consider two cases of investment control - unconstrained investment and investment with a limited amount.
\end{abstract}
\maketitle
\section{Introduction}

Stochastic optimization for insurance business with investment control
 has been studied extensively in recent years.
Under the classical Cramer-Lundberg model (compound Poisson risk process),
 the problem of ruin minimization was first considered in Hipp and Plum \cite{HP}, where the investment control is unconstrained and the investment
amount in the risky asset can be at any level.
More recently, Azcue and Muler \cite{AM} studied the same model
 with a borrowing constraint
that restricts the ratio of borrowed amount to the surplus level to invest in the risky asset.
In Belkina et al. \cite{BHLT}, the authors considered the problem with a restriction that the investment (purchase or short-sell) in the risky asset is allowed only within a limited proportion of the surplus.
In Gaier, Grandits and Schachermayer \cite{GaierGS},
and Hipp and Schmidli \cite{H}, 
for the case with  zero interest and light tailed claims,
 asymptotic behavior of the ruin probability
  was investigated and convergence
of the optimal investment level was proved
 as the surplus tends to infinity.
In Frolova et. al. \cite{kony_FKP}, for the case with exponential claims,
a power function approximation
of the ruin probability was provided under the assumption that
all surplus is invested in the risky asset.
In Gaier and Grandits \cite{GaierG}, \cite{GaierGr}, for
 claims with regularly varying tails, 
and in Grandits \cite{Grandits_r}, 
Schmidli \cite{S11} and Eisenberg \cite{E}, for sub-exponential claims, certain asymptotic properties of the ruin probability and the optimal investment amount
were obtained. Other related research articles worked on more complex control forms with reinsurance and investment.
For example, Schmidli \cite{S1} considered an unconstrained optimal reinsurance-investment control problem under the classical model.
Taksar and Markussen \cite{TM}, Luo \cite{LUO} and Luo et al. \cite{LT} studied the problem under the diffusion approximation model with various investment restrictions.

One common extension of the compound Poisson model
 considers Brownian perturbation, which is known as the jump-diffusion model (See, e.g. \cite{DG}, \cite{GY}, \cite{LIN}, \cite{PG} and \cite{YZ0}).
The surplus process is the sum of the classical risk process and
a Brownian motion.
Under this model, in Zhang and Yang \cite{YZ0},
a general objective function was studied with unconstrained investment control
and numerical methods to compute the optimal investment strategy were discussed.
In Gerber and Yang \cite{GY}, absolute ruin probability was considered.
Laubis and Lin \cite{LL} showed a limiting expression
for the ruin probability function which is a power function
under the assumptions that investment amount is a fixed
fraction of the surplus and that the insurance claims are exponential.
In Lin \cite{LIN}, an exponential
 upper bound for the minimal ruin probability was obtained
and numerical calculations revealing the relationships between the adjustment coefficient and model parameters were conducted.
In a surplus model with stochastic interest rate,
Paulsen and Gjessing \cite{PG} studied the probability
of eventual ruin and the time of ruin without investment control.

This paper studies the problem of ruin probability minimization
with investment control. We consider the perturbed compound Poisson surplus model and assume positive interest rate as in \cite{LL} and \cite{YZ0}.
The surplus can be invested into a risky asset (stock) and a risk-free asset,
where the risky asset price follows a geometric Brownian motion.
We consider two cases of investment.
In the first case, we assume that there is no restriction
on the investment (See, e.g. \cite{HP} and \cite{YZ0}).
That is, the investment amount in the risky asset can be at any level.
Note that short-selling of the risky asset is allowed at any level in this case.
In the second case, we assume that
the investment amount in the risky asset is no more than a fixed level $A$ and short-selling the stock is not allowed.
This restriction is set to reduce leveraging level of the insurer
which was considered in e.g. \cite{LT}.
In both cases, the goal is to minimize the probability of ruin.
The minimal ruin probability function is characterized by an
integro-differential HJB equation as in
\cite{AM}, \cite{HP1}, \cite{HP}, \cite{S2} and \cite{YZ0}.
The HJB equation has a classical solution and
it can be shown that the minimal ruin probability function
is proportional to the solution by a verification result.

We summarize three main contributions of this paper.
First, we define novel operators to prove the existence
of a classical solution to the HJB equation in the two investment cases.
These operators provide an alternative method to compute the optimal
investment strategy and the minimal ruin probability
(See numerical examples in the last section).
Second, we give asymptotic results on the optimal investment strategy
and the minimal ruin function for low surplus levels.
In the unconstrained case, we find that
when the surplus level approaches to $0$, the optimal
investment amount tends to a fixed non-zero value, in contrast to
the model without perturbation in \cite{E} and \cite{HP},
where the optimal investment level tends to zero as surplus tends to $0$.
In addition, we find the rate at which the optimal investment amount
converges to the non-zero level.
Asymptotic results near $0$ for the minimal ruin function are also studied.
In the constrained case, we find close interplay between
the model parameters and the optimal investment control.
We give parameter conditions under which the
optimal investment amount takes values $0$, $A$,
or a certain level in between respectively when the surplus is low.
Note that all these asymptotic results are obtained
for an arbitrary claim size distribution.
Third, in the special case with an exponential claim distribution,
we prove some new asymptotic results for large surplus levels.
We show that the optimal investment
amount has a finite limit when surplus tends to infinity.
We also show that the minimal ruin probability function
has a limiting expression that is the product of
an exponential function and a power function
as surplus tends to infinity.
We note that these new limiting results (of the optimal
investment problem with exponential claims) 
hold also in the classical model without perturbation.   

The rest of this paper is organized as follows.
In Section 2, we formulate the optimization problem.
In Section 3, we prove existence of the classical optimal solution
 and give asymptotic results in the constrained case.
 In Section 4, we study the unconstrained case.
In Section 5, we investigate the model when the claim size is exponential.
Numerical examples and concluding remarks are given in Section 6.

\section{The Optimization Problem}
We assume that without investment the surplus of the insurance company is
governed by the Cramer-Lundberg model:
$$X_t=x + ct-\underset{i=1}{\overset{N(t)}{\sum}}Y_i,$$
where $x$ is the initial surplus,
 $c$ is the premium rate, $N(t)$ is a Poisson process with constant intensity $\lambda$, $Y_i$'s are positive i.i.d. random claims.
Suppose that at time $t$, the insurance company invests an amount of $a_t$ to
a risky asset whose price follows a geometric Brownian motion
$$dS_t=\mu S_tdt+\sigma S_tdB_t,$$
where $\mu$ is the stock return rate, $\sigma$ is the volatility, and $B:=\{B_t\}_{t\geq 0}$
is a standard Brownian motion independent of $\{N(t)\}_{t\geq 0}$ and $Y_i$'s.
The rest of the surplus of amount $(X_t-a_t)$
is invested in a risk free asset which evolves as
$$dP_t=rP_tdt,$$
where $r$ is the interest rate.
We also assume that the surplus process is perturbed by a Brownian noise term.
With perturbation and dynamic investment control, denoted by $\pi:=\{a_s\}_{s\geq 0}$, the surplus process is governed by
\begin{equation}\label{dyn}
\begin{split}
&X^\pi_t=x+\int_0^t[c+r(X_s^\pi-a_s)+\mu a_s]ds+\sigma \int_0^t a_sdB_s + \sigma_1\int_0^t dB^1_s
-\underset{i=1}{\overset{N(t)}{\sum}}Y_i,
\end{split}
\end{equation}
where $B^1:=\{B^1_s\}_{s\geq 0}$ is a standard Brownian motion with $E(dB_sdB^1_s)=\rho ds$ for
some correlation $\rho$ and $\sigma_1$ is the perturbation volatility.
We assume that all the random variables are defined in a complete probability
space $(\Omega, \mathcal{F}, P)$ endowed with a filtration $\{\mathcal{F}_t\}_{t\geq 0}$
generated by processes $\{X_t\}_{t\geq 0}$ and $\{S_t\}_{t\geq 0}$.
A control policy $\pi$ is said to be {\it admissible}
if $a_t$ satisfies the following conditions: (i) $a_t$ is $\mathcal{F}_t$ predictable, (ii) $a_t\in\mathcal{A}$, and (iii) $a_t$ is
square integrable over any finite time interval almost surely,
 where $\mathcal{A}=[0,A]$ for a fixed $A>0$ (constrained case)
 or $\mathcal{A}=(-\infty,\infty)$ (unconstrained case).
We denote by $\Pi$ the set of all admissible controls.
In this paper, we make the following assumptions:
(i) the exogenous parameters $A$, $c$, $r$, $\mu$, $\sigma$, $\sigma_1$ are positive constants;
(ii) $|\rho|\neq 1$ (imperfect correlation);
and (iii) the claim distribution function $F$ has a continuous density
with support $(0,\infty)$ and a finite mean.

Now we define ruin time of the insurance company under admissible policy
$\pi$ as the following
\begin{eqnarray}
\label{ruintime}
\tau^\pi=\inf\{t\geq 0: X^\pi_t<0\}.
\end{eqnarray}
Thus the survival probability under policy $\pi$ is
\begin{eqnarray}
\label{svprob}
\delta^\pi(x)=1-P(\tau^\pi<\infty),
\end{eqnarray}
and the maximal survival probability is
\begin{eqnarray}
\label{opt-svp}
\delta(x)=\underset{\pi\in\Pi}{\sup}\ \delta^\pi(x).
\end{eqnarray}
We see $\delta(x)$ must be a non-decreasing function by its definition.
If we assume that $\delta$ is twice continuously differentiable, then
it solves the following Hamilton-Jacobi-Bellman equation:
\begin{eqnarray}
\label{HJB}
\underset{a\in\mathcal{A}}{\sup}\mathcal{L}^{(a)}\delta(x)=0,
\end{eqnarray}
where the operator $\mathcal{L}$ is given by
\begin{equation}
\label{LM}
\mathcal{L}^{(a)}\delta(x)=\frac{1}{2}(\sigma^2a^2+2\rho\sigma\sigma_1a+\sigma_1^2)\delta''(x)
+[c+(\mu-r)a+rx]\delta'(x)-M(\delta)(x),
\end{equation}
with
\begin{equation}
\label{M}
M(\delta)(x)=\lambda[\delta(x)-\int_0^x\delta(x-s)dF(s)].
\end{equation}
We note that $M(\delta)(x)$ is positive given that
$\delta$ is an increasing function on $(0,\infty)$.
We also note an initial condition $$\delta(0)=0,$$ due to Brownian perturbation.

\section{Asymptotic investment at low surplus - the constrained case}
In this section, we study the optimal control problem in the case when investment control is restricted with the control region $\mathcal{A}=[0,A]$.
We assume $\mu>r$ in this section and omit the other case $\mu\leq r$
which can be treated similarly.

Suppose $W$ is a twice continuously differentiable function and $W$ solves the HJB equation. Write
\begin{equation}
\label{aw}
a_W(x):=-\frac{(\mu-r)W'(x)}{\sigma^2W''(x)}-\rho\frac{\sigma_1}{\sigma},
\end{equation}
when $W''(x)\neq 0$. Then $a_W(x)$ is the unconstrained maximizer
of the right side of HJB equation~\eqref{HJB} when $\mathcal{A}=(-\infty,\infty)$ and if $W''(x)<0$.
The constrained maximizer when $\mathcal{A}=[0,A]$ is given by
\begin{eqnarray}\label{astar-1}
a^*_W(x)=\begin{cases} a_W(x) &\ 0<a_W(x)<A, W''(x)<0 \\
0 &\ a_W(x)\leq 0, W''(x)<0;\ or\ a_W(x)\ge A/2, W''(x)>0\\
A &\ a_W(x)\ge A, W''(x)<0;\ or\ a_W(x) < A/2, W''(x)>0
\end{cases},
\end{eqnarray}
for $W''(x)\neq 0$, and $a^*_W(x)=A$ for $W''(x)=0$.

Now define an operator
\begin{eqnarray}
\label{T}
Tw(x)=\underset{0\leq a\leq A}{\inf}T_w(a,x),
\end{eqnarray}
where
\begin{eqnarray}
\label{Ta}
T_w(a,x)=\frac{2\{M(W)(x)-[c+rx+(\mu-r)a]w(x)\}}{\sigma^2a^2+2\rho\sigma\sigma_1a+\sigma_1^2},
\end{eqnarray}
$w$ is a non-negative and continuous function on $[0,\infty)$, and
$W$ is defined by
\begin{eqnarray}
\label{W}
W(x)=\int_0^xw(s)ds.
\end{eqnarray}
We note $W(0)=0$. Below we give a result that shows existence
of a classical solution to the HJB equation~\eqref{HJB}.
We prove an integral-operator version of
the Picard-Lindel{\" o}f theorem on the initial value problem.
See \cite{T} for the functional version of the arguments.
\begin{lemma}\label{lemmavTv}
There exists a continuously differentiable function $v(x)$ on $[0,\infty)$ satisfying
\begin{equation}\label{vTv}
v'(x)=Tv(x),\ v(0)=1.
\end{equation}
\end{lemma}
\begin{proof}
Notice for any $w\in C[0,K]$ for any $K>0$,
the functions $M(W)(x)$ and $T_w(a,x)$ are continuous in $x$.
Function $T_w(a,x)$ is continuous in $x$ uniformly for $a\in [0,A]$.
Thus function $Tw(x)$ is continuous in $x$.

Now we consider two continuous functions $w_1(x)$ and
$w_2(x)$ in $C[0,K]$.
Use the supremum norm
$$||w||=\sup\{|w(x)|:0 \le x\le K\},$$
for any $w\in C[0,K]$.
Then the following inequalities hold
$$|W_1(x)-W_2(x)|\le \int_0^x |w_1(y)-w_2(y)|dy\le K ||w_1-w_2||,$$
$$\left| \int_0^{x}W_1(x-y)-W_2(x-y)dF(y) \right|\le K||w_1-w_2||\int_0^xdF(y)\le K||w_1-w_2||,$$
where $W_i(x)=\int_0^xw_i(y)dy$ for $i=1,\ 2$.
So we get
$$|M(W_1)(x)-M(W_2)(x)|\leq 2\lambda K||w_1-w_2||.$$
Further, for any $x\in [0,K]$, suppose $Tw_1(x)\geq Tw_2(x)$ and $Tw_2(x)=T_{w_2}(a_2^*,x)$, then
we have
\begin{equation}
\begin{split}
&Tw_1(x)-Tw_2(x)\\
 =&Tw_1(x)-T_{w_2}(a_2^*,x)\\
\leq&T_{w_1}(a_2^*,x)-T_{w_2}(a_2^*,x)\\
=&\frac{2\{M(W_1)(x)-M(W_2)(x)-[c+rx+(\mu-r)a_2^*][w_1(x)-w_2(x)]\}}
{\sigma^2(a_2^*)^2+2\rho\sigma\sigma_1a_2^*+\sigma_1^2}\\
\leq&C(K)||w_1-w_2||,
\end{split}
\end{equation}
where
\begin{equation}
C(K)=2[2\lambda K+c+rK+(\mu-r)A]/\sigma_1^2.
\end{equation}
Thus we see the operator $T$
is Lipschitz with respect to the supremum norm on $C[0,K]$:
\begin{equation}
||Tw_1-Tw_2||\leq C(K)||w_1-w_2||,
\end{equation}
 with Lipschitz constant $C(K)$. Now define operator
\begin{equation}\label{mT}
\mathcal{T}w(x)=\int_0^xTw(y)dy+1,
\end{equation}
for $w\in C[0, K]$. Then it holds
$$||\mathcal{T}w_1-\mathcal{T}w_2||\le KC(K)||w_1-w_2||<1/2||w_1-w_2||,$$
if we select a small $K$ such that $KC(K)<1/2$.
This shows the operator $\mathcal{T}$ is a contraction
on $C[0,K]$. Thus there exists a fixed point $v$ in $C[0,K]$
such that
\begin{equation}\label{vmTv}
v=\mathcal{T}v.
\end{equation}
The solution can be extended to $[0,2K]$ by the same proof and
hence to $[0,\infty)$.
Differentiate \eqref{vmTv} and the lemma is proved.
\end{proof}

Write
\begin{equation}
V(x)=\int_0^xv(s)ds,
\end{equation}
where $v$ is the solution in Lemma~\ref{lemmavTv}. We prove:
\begin{lemma}
$v(x)$ is positive on $[0,\infty)$.
\end{lemma}
\begin{proof}
The lemma can be proved by contradiction. Define
$$x_0=\inf\{x\ge 0: v(x)=0\},$$ and suppose $x_0<\infty$. Then it holds $v(x_0)=0$.
Thus
$$v'(x_0)=Tv(x_0)=
\inf_{0\le a\le A}2M(V)(x_0)/(a^2\sigma^2+2\rho\sigma\sigma_1a+\sigma_1^2)>0,$$
which contradicts:
$$v'(x_0)=\lim_{\epsilon\to 0+}\frac{v(x_0)-v(x_0-\epsilon)}{\epsilon}\le 0.$$
\end{proof}
Now we give a verification lemma:
\begin{lemma}\label{VER}
Suppose $g$ is a positive, increasing, and
twice continuously differentiable function on $[0,\infty)$
and solves the HJB equation~\eqref{HJB}. Then $g$ is bounded and
the maximal survival probability function is given by $\delta(x)=g(x)/g(\infty)$.
Moreover, the associated optimal investment strategy is $\pi^*=\{a^*(t)\}_{t\geq 0}$,
where $a^*(t)=a_{\delta}^*(X^{\pi^*}_{t-})$, and $X^{\pi^*}_t$ is the surplus at time $t$ under
the control policy $\pi^*$.
\end{lemma}
The proof of Lemma~\ref{VER} is similar to the proof of the verification result
in \cite{BHLT} and we omit it. It is easy to check that $V$ solves the HJB equation. Thus by Lemma~\ref{VER}, we obtain:
\begin{theorem}
Function V is bounded and increasing on $[0,\infty)$ and
the maximal survival probability function is given by $$\delta(x)=V(x)/V(\infty).$$
\end{theorem}

\begin{lemma}\label{v'-} There exists $\epsilon>0$ such that $v'(x)<0$ for $0\leq x<\epsilon$.
\end{lemma}
The lemma can be shown by letting $x\rightarrow 0$ in \eqref{vTv}.
The lemma shows local concavity of $V$ at low surplus levels.
In general, concavity of $V$ is not clear in the constrained case.
However, there always exits a concave solution of the HJB equation on $(0,\infty)$ in the unconstrained case (See Section 4).

Write
$$\rho_1=\frac{(\mu-r)\sigma_1}{2c\sigma},\ \rho_2=\rho_1-\frac{(\mu-r)A^2+2cA}{2 c\sigma_1/\sigma}.$$

\begin{theorem}\label{th1}
There exists $\epsilon>0$ such that:
\begin{itemize}
\item[(i)]{if $\rho<\rho_2$, then
$V$ solves
\begin{equation}\label{HJBA}
\sup_{0\le a\le A}\mathcal{L}^aV(x)=\mathcal{L}^AV(x)=0,
\end{equation}
for $x$ in $(0,\epsilon)$;}
\item[(ii)]{if $\rho_2<\rho<\rho_1$, then
$V$ solves
\begin{equation}\label{HJBav}
\sup_{0\le a\le A}\mathcal{L}^{a}V(x)=\mathcal{L}^{a_V(x)}V(x)=0,
\end{equation}
 for $x$ in $(0,\epsilon)$;}
\item[(iii)]{if $\rho>\rho_1$, then
$V$ solves
\begin{equation}\label{HJB0}
\sup_{0\le a\le A}\mathcal{L}^aV(x)=\mathcal{L}^0V(x)=0,
\end{equation}
for $x$ in $(0,\epsilon)$.}
\end{itemize}
\end{theorem}
\begin{proof}
Suppose $\rho<\rho_2$.
Define function
\begin{equation}\label{f}
f(a):=\frac{c+(\mu-r)a}{\frac{1}{2}(\sigma^2a^2+2\rho\sigma\sigma_1a+\sigma_1^2)}.
\end{equation}
By differentiation, function $f$ increases on $(a_1,a_2)$
where
\begin{equation}
a_1,a_2=\frac{-c\pm\sqrt{c^2+(\mu-r)^2\sigma_1^2/\sigma^2-2(\mu-r)\rho c\sigma_1/\sigma}}{\mu-r}.
\end{equation}
Note it holds $a_1<0<a_2$ when $\rho<\rho_1$.
Further, condition $\rho<\rho_2$ is equivalent to $A<a_2$. Thus under the condition
it holds $f(a)\le f(A)$ for $a$ in interval $(0,A)$.
Since $V$ solves HJB equation \eqref{HJB}, we suppose
$\mathcal{L}^{a^*(x)}V(x)=0$ for some $0\leq a^*(x)\leq A$. Then we have
\begin{equation}\begin{split}
v'(x)/v(x)
&=\frac{M(V)(x)/v(x)-[c+(\mu-r)a^*(x)+rx]}{\frac{1}{2}[\sigma^2(a^*(x))^2+2\rho\sigma\sigma_1a^*(x)+\sigma_1^2]}\\
&=\frac{M(V)(x)/v(x)-rx}{\frac{1}{2}[\sigma^2(a^*(x))^2+2\rho\sigma\sigma_1a^*(x)+\sigma_1^2]}-f(a^*(x))\\
&\ge \frac{M(V)(x)/v(x)-rx}{\frac{1}{2}[\sigma^2(a^*(x))^2+2\rho\sigma\sigma_1a^*(x)+\sigma_1^2]}-f(A).
\end{split}
\end{equation}
So it holds
\begin{equation}\label{eq1}
a_V(x)\ge \frac{\mu-r}{\sigma^2}\frac{1}{f(A)-\frac{M(V)(x)/v(x)-rx}{\frac{1}{2}[\sigma^2(a^*(x))^2+2\rho\sigma\sigma_1a^*(x)+\sigma_1^2]}}-\rho\sigma_1/\sigma>A,
\end{equation}
when $x$ is sufficiently small. The second inequality is because $\lim_{x\rightarrow 0}M(V)(x)=0$,
$\lim_{x\rightarrow 0}v(x)=1$ and
the fact that the parameter condition $\rho<\rho_2$
is equivalent to $$\frac{\mu-r}{\sigma^2f(A)}-\rho\sigma_1/\sigma>A.$$
Also, from Lemma~\ref{v'-}, we have $v'(x)<0$ when $x$ is small. Thus from \eqref{eq1} ($a_V(x)>A$),
it holds $$\sup_{0\le a\le A}\mathcal{L}^aV(x)=\mathcal{L}^AV(x),$$
when $x$ is small.
Since $V$ solves the HJB equation and it holds $\sup_{0\le a\le A}\mathcal{L}^aV(x)=0$, this proves \eqref{HJBA}.

Now we suppose $\rho>\rho_1$ and prove (iii). From $\mathcal{L}^0V(x)\le 0$, we have
\begin{equation}
v'(x)/v(x)\le \frac{M(V)(x)/v(x)-(c+rx)}{\frac{1}{2}\sigma_1^2}.
\end{equation}
Then it holds
\begin{equation}\label{eq2}
a_V(x)\le \frac{\mu-r}{\sigma^2}\frac{\frac{1}{2}\sigma_1^2}{c-[M(V)(x)/v(x)+rx]}-\rho\sigma_1/\sigma<0,
\end{equation}
for small $x$, where the second inequality is due to parameter condition $\rho>\rho_1$.
Thus it holds
$$\sup_{0\le a\le A}\mathcal{L}^aV(x)=\mathcal{L}^0V(x),$$
for small $x$. This proves \eqref{HJB0}.

Under condition $\rho_2<\rho<\rho_1$, we prove \eqref{HJBav} by contradiction.
Suppose the maximizer of $\mathcal{L}^aV(x)$ in $a$ is not $a_V(x)$ but $0$ or $A$.
Then HJB equation \eqref{HJBA} or \eqref{HJB0} holds.
If it holds $\mathcal{L}^AV(x)=0$, we have
$$a_V(x)=\frac{\mu-r}{\sigma^2}\frac{1}{f(A)-
\frac{M(V)(x)/v(x)-rx}{\frac{1}{2}(\sigma^2A^2+2\rho\sigma\sigma_1A+\sigma_1^2)}}
-\rho\sigma_1/\sigma<A,$$
for small $x$ due to $\rho>\rho_2$. Thus the maximizer of $\mathcal{L}^aV(x)$ in $a$
is not $A$ and $\mathcal{L}^AV(x)<0$. Contradiction!

If $\mathcal{L}^0V(x)=0$, we have
$$a_V(x)=\frac{\mu-r}{\sigma^2}\frac{\frac{1}{2}\sigma_1^2}{c-[M(V)(x)/v(x)+rx]}-\rho\sigma_1/\sigma>0,
$$
for small $x$ due to $\rho<\rho_1$. Thus the maximizer of $\mathcal{L}^aV(x)$ in $a$
is not $0$ and $\mathcal{L}^0V(x)<0$. Contradiction!

Hence under parameter condition $\rho_2<\rho<\rho_1$, the maximizer of $\mathcal{L}^aV(x)$ in $a$
is $a_V(x)$ for small $x$ and it holds \eqref{HJBav}.
\end{proof}

By Theorem~\ref{th1}, we obtain the following properties:
\begin{remark}
At low surplus level $x$, the optimal investment control $a_V^*$, is $a_V^*(x)=A$ if $\rho<\rho_2$,
 $a_V^*(x)=a_V(x)$ if $\rho_2<\rho<\rho_1$, and $a_V^*(x)=0$ if $\rho>\rho_1$.
Moreover, when $\rho_2<\rho<\rho_1$, letting $x\rightarrow 0$ in \eqref{vTv}, $a_V^*(0+)$ is equal to $a_V(0+)$
and solves the following equation:
 $$-\frac{1}{\frac{\sigma^2}{\mu-r}(a+\rho\frac{\sigma_1}{\sigma})}
=\frac{-2[c+(\mu-r)a]}{\sigma^2a^2+2\rho\sigma\sigma_1a+\sigma_1^2},$$
 which simplifies to $$(\mu-r)\sigma^2a^2+2c\sigma^2a+[2\rho\sigma\sigma_1c-\sigma_1^2(\mu-r)]=0.$$
 Thus \begin{equation}\label{astar0}
a_V^*(0+)=-\frac{c}{\mu-r}+\sqrt{\frac{c^2}{(\mu-r)^2}+\frac{2\sigma_1c}{\sigma(\mu-r)}(\rho_1-\rho)}.
 \end{equation}
Note it holds $0<a_V^*(0+)<A$
under the parameter condition $\rho_2<\rho<\rho_1$.
If $\mu\neq r$, the quadratic equation above has two real roots.
\end{remark}
\begin{remark}
$v'(0+)=-\frac{2[c+(\mu-r)A]}{\sigma^2A^2+2\rho\sigma\sigma_1A+\sigma_1^2}$ if $\rho<\rho_2$;
$v'(0+)=-\frac{2c}{\sigma_1^2}$ if $\rho>\rho_1$;
and $v'(0+)=-\frac{\mu-r}{\sigma^2[a_V^*(0+)+\rho\frac{\sigma_1}{\sigma}]}
=-\frac{\mu-r}{\sigma^2[\rho\frac{\sigma_1}{\sigma}
-\frac{c}{\mu-r}+\sqrt{\frac{c^2}{(\mu-r)^2}
+\frac{2\sigma_1c}{\sigma(\mu-r)}(\rho_1-\rho)}]}$ if $\rho_2<\rho<\rho_1$.
Note it holds $v'(0+)<0$ in all the cases.
\end{remark}

\section{Asymptotic investment at low surplus - the unconstrained case}
In this section, we consider the unconstrained case. That is,
the control region is given by $\mathcal{A}=(-\infty,\infty)$.
We assume $\mu\neq r$ in the derivation part of this section. 
The special case with $\mu=r$ is addressed in Remark~\ref{rmk-mr}.

Let us assume that the HJB equation \eqref{HJB} has a classical bounded solution $V(x)$ which satisfies
  \begin{equation}
  \label{incrconc}
  V'(x)>0,\ V''(x)<0,
\end{equation}
and
\begin{equation}
\label{v0}
V(0) = 0,
\end{equation}
\begin{equation}
\label{v'0}
V'(0+)=1.
\end{equation}
Then the maximizer in \eqref{HJB} is
\begin{equation}
\label{optimum_nonconstr}
 a^*_V(x)=a_V(x),
\end{equation}
 and $V$ solves
\begin{equation}\label{HJBnonc}
\sup_{-\infty< a< \infty}\mathcal{L}^{a}V(x)=\mathcal{L}^{a_V(x)}V(x)=0,
\end{equation}
i.e.,
\begin{equation}
\label{noconstrain}
\left[{c + rx-\frac{\rho(\mu  - r)\sigma _1}{\sigma}}\right]V'\left( x
\right) + \frac{1}{2}\sigma _1^2(1-\rho ^2) V''\left( x \right) +
\lambda E\left[ {V\left( {x - Y} \right) - V\left( x \right)}
\right]
=\gamma\frac{{\left( {V'\left( x \right)} \right)^2 }}{{V''\left( x \right)}}.
\end{equation}
where
\begin{equation}\label{gamma}
\gamma=\frac{{\left(\mu  - r\right)^2 }}{{2\sigma ^2}}.
\end{equation}
\begin{remark}
It is easy to see that in the unconstrained case,
the limit optimal investment amount at $0+$ is $a_V^*(0+)=a_V(0+)$, which is given in \eqref{astar0}. It holds  $a_V^*(0+)>0$ under the parameter condition $\rho<\rho_1$, $a_V^*(0+)=0$ if $\rho=\rho_1$, and  ${a_V^*(0+)<0}$ if $\rho>\rho_1$.
This corrects a mistaken assumption on page 624 of \cite{YZ0} that
the optimal initial investment amount is always $0$.
\end{remark}
Below we proceed to show that there exists a classical solution $V$
that solves the HJB equation \eqref{HJBnonc}
with respect to conditions \eqref{incrconc}, \eqref{v0} and \eqref{v'0}.

Let  $H(y) = 1 - F(y)$. Then taking into account condition
\eqref{v0}, the equation \eqref{noconstrain} can be rewritten in
the form
  \begin{equation}
  \label{newform}
\left[c + rx - \frac{\rho(\mu  - r)\sigma _1}{\sigma}\right]V'(x) +
\frac{1}{2}\sigma _1^2 (1-\rho ^2)V''(x) - \lambda \int\limits_0^x
H(y)V'(x - y)dy
=\gamma\frac{{\left( {V'(x)} \right)^2
}}{{V''(x)}},
\end{equation}
or
\begin{equation}\label{newform1}
\frac{1}{2}\sigma _1^2 (1-\rho ^2)\left [V''(x)\right ]^2+L_1V'(x)V''(x)-\frac{1}{2}\left[L_2V'(x)\right]^2=0,
\end{equation}
where operators $L_1$ and $L_2$ are defined as the following:
\begin{equation}
\begin{split}
&L_1w(x)=\left[c + rx - \frac{\rho(\mu  - r)\sigma _1}{\sigma}\right]w(x) -
\lambda \int\limits_0^x H(y)w(x - y)dy,\\
&L_2w(x)=\frac{(\mu- r)}{\sigma}w(x),
\end{split}
\end{equation}
for any continuous function $w$ on $(0,\infty)$.
Then from \eqref{newform1} and \eqref{incrconc}, $V''$ satisfies
\begin{equation*}
V''(x)=LV'(x),
\end{equation*}
with condition $V''\le 0$, where the operator $L$ is defined by:
\begin{equation}\label{L}
Lw(x)=-\frac{L_1w(x)+\sqrt{(L_1w(x))^2
+\sigma _1^2(1-\rho^2)(L_2w(x))^2}}{\sigma _1^2(1-\rho^2)}.
\end{equation}

\begin{lemma}
There exists a
continuously differentiable function $v(x)$ on $(0,\infty)$ satisfying
\begin{equation}\label{vLv}
v'(x)=Lv(x),\ v(0)=1.
\end{equation}
\end{lemma}
\begin{proof}
Notice that for any $w\in C[0,K]$ for any $K>0$,
the functions $L_1w$, $L_2w$ and $Lw$ are continuous in $x$.
Consider two continuous functions $w_1(x),w_2(x)$ in $C[0,K]$ and
the supremum norm
$$||w||=\sup\{|w(x)|:0 \le x\le K\}.$$
Then the following inequalities hold:
$$|L_1w_1(x)-L_1w_2(x)|\le [c+rK+(\mu-r)\sigma_1/\sigma+\lambda K]||w_1-w_2||,$$
$$|L_2w_1(x)-L_2w_2(x)|\le (\mu-r)/\sigma||w_1-w_2||.$$
Also notice that function $f(x,y)=\sqrt{x^2+y^2}$ is Lipschitz:
$$|f(x_1,y_1)-f(x_2,y_2)|\le |x_1-x_2|+|y_1-y_2|.$$
Then we see:
\begin{equation}\begin{split}
&|Lw_1(x)-Lw_2(x)|\\
\leq&\frac{|L_1w_1(x)-L_1w_2(x)|}{\sigma_1^2(1-\rho^2)}
+\frac{|f(L_1w_1(x),\sigma_1\sqrt{1-\rho^2}L_2w_1(x))-
f(L_1w_2(x),\sigma_1\sqrt{1-\rho^2}L_2w_2(x))|}{\sigma_1^2(1-\rho^2)}\\
\leq& \frac{1}{\sigma_1^2(1-\rho^2)}
[2||L_1w_1-L_1w_2||+\sigma_1||L_2w_1-L_2w_2||]\\
\leq& C||w_1-w_2||,
\end{split}
\end{equation}
where
\begin{equation}
C=\frac{2c+2rK+3(\mu-r)\sigma_1/\sigma+2\lambda K}{\sigma_1^2(1-\rho^2)}.
\end{equation}
Thus the operator $L$ is Lipschitz on $C[0,K]$ with constant $C$:
\begin{equation}
||Lw_1-Lw_2||\leq C||w_1-w_2||.
\end{equation}
The rest of the proof is similar to that of Lemma~\ref{lemmavTv}
and we skip it.
We then conclude \eqref{vLv}.
\end{proof}

Write
\begin{equation}\label{V-2}
V(x)=\int_0^xv(y)dy,
\end{equation}
where $v$ is the solution of \eqref{vLv}.
We prove the following lemma:
\begin{lemma}
$v(x)$ is positive and $v'(x)$ is negative on $[0,\infty)$.
\end{lemma}
\begin{proof}
We prove by contradiction. Define
$$x_0={\inf}\{x>0: v(x)=0\},$$
and suppose $x_0<\infty$. Then it holds $v(x)>0$ on $[0,x_0)$,
 $v(x_0)=0$ and $v'(x_0)=0$
 (due to continuity of $v$
 and definition of $L$). It also holds $v'(x)<0$ for $x\in [0,x_0)$.

Notice that $V$ satisfies \eqref{newform} on $(0,x_0)$.
Passing $x\uparrow x_0$ in \eqref{newform}, we obtain
$$\underset{x\uparrow x_0}{\lim}
 \frac{{\left(\mu  - r\right)^2 }}{{2\sigma ^2 }}\frac{{\left( {V'(x)} \right)^2
}}{{V''(x)}}= -\lambda\int\limits_0^{x_0}H(y)V'(x_0 - y)dy,$$
wherefrom we see $\underset{x\uparrow x_0}{\lim}
\frac{V'(x)}{ V''(x)/V'(x)}$ is finite and hence
$$\underset{x\uparrow x_0}{\lim}
\frac{V''(x)}{ V'(x)}=0.$$
Choose $x_1$ and $x_2$ close to $x_0$ with $0<x_2<x_1<x_0$ such that
$-\frac{v'(x)}{v(x)}<1$ for $x\in (x_2,x_0)$. Thus
\begin{equation}
\ln v(x_2)-\ln v(x_1)=\int_{x_2}^{x_1}-(\ln v(x))'dx=\int_{x_2}^{x_1}-\frac{v'(x)}{v(x)}dx
<x_1-x_2.
\end{equation}
Passing $x_1\uparrow x_0$ in the above, contradiction!
Thus it must hold $x_0=\infty$ and we conclude $v(x)$ is
positive on $[0,\infty)$. Immediately, $v'(x)$ is negative.
\end{proof}
To this end, we see that the function $V$ given by \eqref{V-2}
solves the HJB equation \eqref{HJBnonc}
with respect to conditions \eqref{incrconc}, \eqref{v0} and \eqref{v'0}.
By a verification result (similar to Lemma~\ref{VER}),
we see that the maximal survival function (value function) is given by
$\delta(x)=V(x)/V(\infty)$.

Next we study the asymptotic behavior of the HJB solution
for low initial surplus.
Write
$$\tilde{v}(x) =v(x)- 1,$$
and
\begin{equation}\label{crho}
c_\rho=c-\frac{\rho(\mu - r)\sigma _1}{\sigma},\
\sigma_\rho^2=\sigma_1^2(1-\rho ^2),
\end{equation}
then  $\tilde{v}(x)$   satisfies equation
\begin{equation*}\label{newform2}
(c_\rho+rx)(\tilde{v}(x) +1) + \frac{1}{2}\sigma_\rho^2\tilde{v}'(x)
-\lambda \int\limits_0^x H(y)\left( {\tilde{v}(x - y) + 1} \right)dy
=\gamma \frac{{\left( {\tilde{v}(x) + 1} \right)^2 }}{{\tilde{v}'(x)}},
\end{equation*}
where  $\gamma$ is defined in \eqref{gamma}, wit initial condition
\begin{equation}
\label{tildev'0}
\tilde{v}(0+)=0.
\end{equation}
Multiplying both sides of the last equation by $\tilde{v}'(x)$,
we get equation
\begin{equation} \label{tilde v}\begin{split}
&(c_\rho+rx)(\tilde{v}(x) +1)\tilde{v}'(x) + \frac{1}{2}\sigma_\rho^2(\tilde{v}'(x))^2\\
&-\lambda \tilde{v}'(x)\int\limits_0^xH(y)\left( {\tilde{v}(x - y) + 1} \right)dy = \gamma {{\left( {\tilde{v}(x) + 1} \right)^2 }}.
\end{split}
\end{equation}
 We find  representations of the solution of equation \eqref{tilde v} with condition \eqref{tildev'0} and its derivative in such forms as:
 $$\tilde{v}(x) = \alpha x^\beta  (1 + o(1)),\
{\tilde{v}'(x) = \beta \alpha x^{\beta-1}(1 + o(1))},\ x \to 0, $$
where  $\beta >0$ and $\alpha$  are some constants. Taking into account that
   ${H(x) = 1 + o(1)}$, $x \to 0$, we have from \eqref{tilde v} on principal terms of expansion:
\begin{equation*}
\begin{split}
&\left[{\gamma\alpha^2x^{2\beta}+2\gamma\alpha x^\beta+\gamma}\right]
\left({1+ o(1)} \right)\\
=&\left[ {c_\rho\alpha^2 \beta x^{2\beta-1}  + c_\rho\alpha\beta x^{\beta  - 1}  + \left( {r - \frac{\lambda }{{\beta  + 1}}} \right)\alpha^2 \beta x^{2\beta }   } \right. \\
 &+\left. {(r - \lambda )\alpha\beta x^\beta
 + \frac{1}{2}\sigma_\rho^2\alpha^2 \beta ^2 x^{2\beta-2}}\right]
\left( {1 + o(1)}\right),
\,\,\,x \to 0.
\end{split}
 \end{equation*}
From this relation it is easy to see that
$\beta = 1.$   Then
$$c_\rho\alpha+\frac{1}{2}\sigma_\rho^2\alpha^2=\gamma.$$
Therefore
\begin{equation*}
\begin{split}
\alpha &= \tilde{v}'(0+)\\
&=-\left(c_\rho+\sqrt{c_\rho^2+2\gamma\sigma_\rho^2}\right)/{\sigma_\rho^2}\\
&=-\frac{\mu-r}{\sigma^2[\rho\frac{\sigma_1}{\sigma}
-\frac{c}{\mu-r}+\sqrt{\frac{c^2}{(\mu-r)^2}
+\frac{2\sigma_1c}{\sigma(\mu-r)}(\rho_1-\rho)}]}.
\end{split}
\end{equation*}
Recall that ${\tilde{v}}'(0+)= V''(0+) < 0$  in view of \eqref{incrconc}.
We have,
   \begin{equation}
   \label{V'(x)}
V'\left( x \right) = 1 -  Bx(1 + o(1)),\,\,\,x \to 0,
\end{equation}
and
 \begin{equation}
 \label{V''(x)}
V''\left( x \right) =  -  B + o\left( 1 \right),\,\,\,x \to 0,
\end{equation}
 where
 \begin{equation}\label{B}
 B=\left(c_\rho+ \sqrt {c_\rho^2  +2\gamma\sigma_\rho^2}\right)/{\sigma_\rho^2}.
 \end{equation}
Note that
$$B=\frac{\mu-r}{\sigma^2[a_V^*(0+)+\rho\frac{\sigma_1}{\sigma}]},$$
where $a_V^*(0+)$ is given by \eqref{astar0}.
Hence, we get asymptotic representation of $V(x)$:
$$V(x)={x-(B/2)x^2(1 + o(1))},\ x \to 0.$$
In view of the value function $\delta(x)$, we have
\begin{theorem}\label{thm41}
It holds $$\delta(x)=C[x-(B/2)x^2 (1 + o(1))] ,\,\,\,x \to 0,$$
where $B$ is given in \eqref{B} and $C=1/V(\infty)$ is a positive constant.
\end{theorem}
 Next, we find more exact asymptotic representation
of $ V''(x)$ to obtain an asymptotic representation of the optimal strategy.
For this, we introduce the change of variables
$$ \tilde{v}(x) = -B x  (1 + z(x)),$$
where $B$ is defined in \eqref{B}. Then
$${\tilde{v}'(x) = -B  (1 + z(x))}-B x z'(x),\ \ x \to 0.$$
We characterize $z(x)$ in the following form:
$$z(x)=\eta x^{\theta}(1+o(1)),$$
 where  $\theta >0$ and $\eta$  are some constants.
From \eqref{tilde v} we have $\theta=1$ and then
  $$\eta =\frac{\lambda -r +2\gamma+Bc_{\rho}}{2(c_{\rho}-B\sigma_{\rho}^2)},$$
 and
\begin{equation}
\label{V''(x)new}
V''(x)=-B-2B\eta x(1+o(1)), \,x \to 0.
\end{equation}
For the optimal
strategy \eqref{optimum_nonconstr}, in view of \eqref{aw},
\eqref{V'(x)} and \eqref{V''(x)new},
we obtain for $x \to 0$:
\begin{equation}
\begin{split}
\label{a*zero}
a^*_V(x)&=-\frac{(\mu-r)V'(x)}{\sigma^2V''(x)}-\rho\frac{\sigma_1}{\sigma}\\
&=\frac{(\mu-r)[1-Bx(1+o(1))]}{\sigma^2B[1+2\eta x(1+o(1))]}-\rho\frac{\sigma_1}{\sigma}\\
&=-\frac{c}{\mu-r}+\sqrt{\frac{c^2}{(\mu-r)^2}+\frac{2\sigma_1c}{\sigma(\mu-r)}(\rho_1-\rho)}
-\frac{\mu
-r}{\sigma ^2}\left(1+\frac{2\eta}{B}\right)x(1 + o(1)).
\end{split}
\end{equation}
Finally, we have
\begin{theorem}\label{thm42}
For the optimal investment strategy, it holds
$$ a^*_V(x)=a_V^*(0+)
-\left(\frac{\mu
-r}{\sigma ^2}-\frac{\left[\lambda -r+ \frac{{\left( {\mu  - r} \right)^2 }}{{\sigma ^2 }}\right]\left[a_V^*(0+)+\rho\frac{\sigma_1}{\sigma}\right]+c_\rho \frac{{\left( {\mu  - r} \right) }}{{\sigma ^2 }}}{\sqrt{c_\rho ^2+\frac{{\left( {\mu  - r} \right)^2 }}{{\sigma ^2 }}\sigma _\rho ^2}}\right)x(1 + o(1)),$$
when $x \to 0$, where $a^*_V(0+)$ is given in \eqref{astar0}, $c_{\rho}$ and $\sigma_{\rho}$ are given in \eqref{crho}.
\end{theorem}
\begin{remark}
We note that the results in Theorems~\ref{thm41} and ~\ref{thm42}
hold also in the constrained case under the parameter condition
$\rho_2<\rho<\rho_1$.
These results are obtained for any claim size distribution
with the property $H(x)=1-o(1)$ when $x\rightarrow 0$.
\end{remark}
\begin{remark}\label{rmk-mr}
The results in Theorems~\ref{thm41}
and ~\ref{thm42} hold for the parameter case $\mu=r$ under which
the optimal investment amount is a constant strategy with
$a^*_V(x)\equiv-\rho\frac{\sigma_1}{\sigma}$ and $B=2c_\rho/\sigma_\rho^2$.
\end{remark}
\begin{remark} In the case with unconstrained investment,
when $\mu-r>0$ and $\rho>\rho_1$, we have $a_V^*(0+)<0$,
which indicates that it is optimal to short-sell the high return stock
to earn interest at low surplus levels.
We note that this counter-intuitive investment strategy, which
never occurs in the classical 
model without perturbation in the unconstrained case,
shows a special feature of the perturbed model that
investment (buying/short-selling the stock)
is not only for the stock return
but also for neutralizing the perturbation risk.
 We also note that
 this strategy (short-selling the high return
stock to earn interest) can occur when a strong investment 
 constraint on borrowing (money) and buying (stock) is imposed
in the model without perturbation (See, e.g. \cite{BHLT}).
\end{remark}

\section{Analysis of the case of exponential claims}
We now analyze the case when the claim size
 $Y$ has an exponential
distribution with mean $m$. In this case,
we show some new results on asymptotic behaviors for large surplus levels.
For the special case $\mu=r$, from $a_V(x)\equiv-\rho\frac{\sigma_1}{\sigma}$,
 we see that the optimal investment amount is a constant and this case
is addressed in Remark~\ref{rmk-54}.
In the following, we assume $\mu\neq r$.
\subsection{The unconstrained case}
As in \cite{Norshteyn} (the case $\rho =0$),
we first derive an equation for the optimal strategy.
From equation \eqref{newform}, for the case with exponential claim
 distribution function $H\left(y\right)= e^{-ky}$, where $k=1/m$,
we have
\begin{equation}\label{newform3}
(c_\rho+rx){v}(x)+
\frac{1}{2}\sigma_\rho^2{v}'(x)
 -\lambda \int\limits_0^x e^{-ky}v(x-y)dy - \lambda V(0)e^{-kx}
=\gamma \frac{{\left( {{v}(x)} \right)^2 }}{{{v}'(x)}},
\end{equation}
where  $v(x)=V'(x)$, $\gamma$ is given in \eqref{gamma}, $c_\rho$ and $\sigma_\rho$ are given in \eqref{crho}.
 Let $$u(x) = v(x)e^{kx}.$$
It holds $ v'(x)=e^{- kx} (u'\left( x \right)-ku\left(x\right))$.
Then equation \eqref{newform3} can be rewritten as
\begin{equation*}
(c_\rho+rx)u(x)
 + \frac{1}{2}\sigma_\rho^2(u'(x)-ku(x))
-\lambda \int\limits_0^x u\left( y \right)dy -  \lambda V(0)
 = \gamma \frac{{u^2
\left( x \right)}}{u'(x) - ku(x)},
\end{equation*}
and
 $$ {a^*_V(x) =
-\frac{(\mu-r)V'(x)}{\sigma^2V''(x)}-\rho\frac{\sigma_1}{\sigma} =
 -\frac{{\mu  - r}}{{\sigma ^2 }}\frac{{u\left( x
\right)}}{{u'\left( x \right) - ku\left( x
\right)}}-\rho\frac{\sigma_1}{\sigma}}.$$
Hence
\begin{equation*}
\begin{split}
&(c_\rho+rx)u(x) - \frac{1}{2}\sigma _\rho^2
\frac{{\mu-r}}{\left(a^*_V(x)+\rho\frac{\sigma_1}{\sigma}\right)
\sigma ^2}u\left( x \right)
-\lambda \int\limits_0^x u\left( y \right)dy -  \lambda V(0)\\
=& -\gamma \frac{(a^*_V(x)+\rho\frac{\sigma_1}{\sigma})\sigma ^2
}{\mu - r}u(x).
\end{split}
\end{equation*}
Differentiating this equation with respect to $x$, we get
\begin{equation*}
\begin{split}
& ru(x) + (c_\rho+rx)u'(x)-\frac{1}{2}\sigma _\rho^2\frac{{\mu  - r}}
 {{\sigma ^2 }}\frac{{u'(x)(a^*_V(x)+\rho\frac{\sigma_1}{\sigma}) - u\left( x \right)(a^*_V(x))'}}
 {{(a^*_V(x)+\rho\frac{\sigma_1}{\sigma})^2 }}-\lambda u\left( x \right)\\
 = &- \frac{{\mu  - r}}{2}\left[
{u'\left( x \right)(a^*_V(x)+\rho\frac{\sigma_1}{\sigma}) +
  u\left( x \right)(a^*_V(x))'} \right].
\end{split}
\end{equation*}
Divide both sides by $u(x)$, and write
\begin{equation}
\label{tilde_a}
\tilde{a}_V(x)=a^*_V(x)+\rho\frac{\sigma_1}{\sigma}
=-\frac{\mu-r}{\sigma^2}\frac{V'(x)}{V''(x)},
\end{equation}
we have the following equation for $\tilde{a}_V$:
\begin{equation*}
\begin{split}
 &r-\lambda +(c_\rho+rx)\left({k-\frac{{\mu  - r}}{{\sigma ^2\tilde{a}_V(x)}}}\right)
 -\frac{1}{2}\sigma_\rho^2\frac{{\mu  - r}}{{\sigma ^2 }}
 \frac{{k\tilde{a}_V(x) - \frac{{\mu  - r}}{{\sigma ^2 }}
 - \tilde{a}_V'(x)}}{{\tilde{a}_V^2(x) }}\\
 =& -\frac{{\mu  - r}}{2}\left[ {k\tilde{a}_V(x)-\frac{{\mu- r}}{{\sigma ^2 }} + \tilde{a}_V'(x)} \right],
\end{split}
\end{equation*}
and finally,
\begin{equation}
 \label{eq_tilde_a}
\begin{split}
[\sigma ^2 \tilde{a}_V^2(x)+\sigma _\rho^2]\tilde{a}_V'(x)
=&-\frac{\sigma ^2}{m} \tilde{a}_V^3(x)-2\left[r-\lambda+ \frac{c_\rho}{m} -\frac{{(\mu-r)^2 }}{2\sigma ^2} + \frac{r}{m}x \right]
\frac{\sigma ^2}{\mu-r}\tilde{a}_V^2(x)\\
&+2\left(c_{\rho} + rx + \frac{1}{2m}\sigma _\rho^2\right)\tilde{a}_V(x)
  - \sigma _\rho^2 \frac{{\mu  - r}}{{\sigma ^2 }}.
\end{split}
 \end{equation}

In the sequel, we find asymptotic representations of the optimal strategy and the value function at infinity.
It can be shown that equation \eqref{eq_tilde_a} has a family of bounded solutions, each of which is representable in principal in the form of the following asymptotic series for large $x$:
\begin{equation}
\label{series} \tilde a_V(x)\sim\Sigma
_{k=0}^{\infty}\tilde{a}_kx^{-k},
\end{equation}
 where
$$\tilde{a}_0=\frac{(\mu-r)m}{\sigma ^2}, \
\tilde{a}_1=-\left(1-\frac{\lambda}{r}\right)\frac{(\mu-r)m^2}{\sigma ^2}.$$
A short justification of \eqref{series}, for the function $\tilde a_V$ defined in \eqref{tilde_a}, is given in the Appendix.
We then obtain the following property of the optimal strategy:
 \begin{theorem} \label{th51}
It holds
 \begin{equation}
 \label{asymp_a*}
 a^*_V(x)= \frac{(\mu-r)m}{\sigma ^2}-\rho\frac{\sigma_1}{\sigma}-\left(1-\frac{\lambda}{r}\right)\frac{(\mu-r)m^2}{\sigma ^2}\frac{1}{x}(1+o(1)),
\end{equation}
for  $x\rightarrow \infty$.
\end{theorem}

Noticing $\left[\ln{V'(x)}\right]'=-\frac{\mu-r}{\sigma^2\tilde{a}_V(x)}$,
 we have
\begin{equation}
\label{Vint}
V'(x)=K\exp{\{-\frac{(\mu-r)}{\sigma ^2}\int _0^x \frac{1}{\tilde{a}_V(y)}dy\}},
\end{equation}
for some $K>0$, wherefrom using $\tilde{a}_V(y)=\tilde{a}_0+\tilde{a}_1\frac{1}{y}(1+o(1))$ for large $y$,
we get
\begin{equation*}
\begin{split}
V'(x)=&K\exp{\{-\int _1^x \frac{1}{m\left[1-m\left(1-\frac{\lambda}{r}\right)\frac{1}{y}(1+o(1))\right]}dy\}}\\
=&K\exp{\{-\int _1^x\left[\frac{1}{m}+\left(1-\frac{\lambda}{r}\right)\frac{1}{y}(1+o(1))\right]dy\}}\\
=&Ke^{-{x}/{m}}x^{{\lambda}/{r}-1}(1+o(1)),
\end{split}
\end{equation*}
for $x\rightarrow \infty$. Thus it holds
 $$V(x)= V(\infty)-Ke^{-{x}/{m}}x^{{\lambda}/{r}-1}(1+o(1)),\ x\rightarrow \infty,$$
 and we have the relation on the value function:

 \begin{theorem}\label{th52} It holds
$$\delta(x)=1-K_1e^{-{x}/{m}}x^{{\lambda}/{r}-1}(1+o(1)),\ x\rightarrow \infty,$$
for some constant $K_1>0$.
\end{theorem}
Write the minimal ruin probability function as:
$$\Psi(x)=1-\delta(x),$$
and we have
\begin{eqnarray}\label{ruinbd}
\Psi(x)=K_1e^{-{x}/{m}}x^{{\lambda}/{r}-1}(1+o(1)),\ x\rightarrow \infty.
\end{eqnarray}

\begin{remark}
We note that the results in \eqref{asymp_a*} and \eqref{ruinbd} also
hold in the model without perturbation $\sigma_1=0$.
From \eqref{asymp_a*}, we see the optimal investment amount has a finite limit.
The rate at which the optimal strategy converges to the limit is also given.
In \eqref{ruinbd},
the limiting expression of the minimal ruin probability function is
a product of an exponential function and a power function as
$x\rightarrow\infty$. We see that the interest rate $r$,
the exponential claim mean $m$ and
the claim occurrence intensity $\lambda$ play key roles in the expression, while the stock parameters $\mu$, $\sigma$,
 the premium rate $c$ and the perturbation parameter $\sigma_1$ are insignificant.
\end{remark}
To compare \eqref{ruinbd} with the existing results
(e.g., exponential bounds or power function approximation
 for the ruin probability function),
we give the following remarks:
\begin{remark}
In Theorem 1 of \cite{kony_FKP}, it is shown if $\frac{2\mu}{\sigma^2}>1$,
$\Psi(x)=Kx^{1-\frac{2\mu}{\sigma^2}}(1+o(1)),\ x\rightarrow\infty,$
 for some constant $K>0$.
This result is obtained under the model without Brownian perturbation. And there is no control of investment, i.e., with all the surplus invested in the risky asset. One extension of this result is given in Theorem 2 of \cite{LL}, where
the surplus process is modeled by the perturbed Cramer-Lundberg process. With a constant proportion of investment, i.e. the investment amount at time $t$ is
 $a_t=\alpha X_t$, for $0<\alpha\leq 1$, the ruin probability takes the form:
$$\Psi(x)=Kx^{1-\frac{2[\mu\alpha+r(1-\alpha)]}{\alpha^2\sigma^2}}(1+o(1)),
\ x\rightarrow\infty,$$
for some constant $K$ and $\frac{2[\mu\alpha+r(1-\alpha)]}{\alpha^2\sigma^2}>1$.
 In these models, investment amount in the risky asset tends to infinity as
surplus tends to infinity.
Consequently, when the surplus level is large,
 the stock volatility and stock growth
are major parameters that affect the ruin probability,
but not the exponential mean, claim occurrence intensity and
 premium rate.
\end{remark}
\begin{remark}
In Theorem 4.1 of \cite{GaierGS}, it is shown $\Psi(x)\leq e^{-Rx}$ where $R\in(0,1/m)$ solves the equation
$$\lambda\left(\frac{1}{1-mR}-1\right)=cR+\frac{\mu^2}{2\sigma^2}.$$
We note that the surplus process there is a special case
of the jump-diffusion process in this paper with $r=0$ and $\sigma_1=0$,
 and that he bound is obtained by using a constant investment policy $\pi$ with amount $\frac{\mu}{R\sigma^2}$
under which the process $\{e^{-RX^\pi_t}\}_{t\ge 0}$ is a martingale.
In Theorem 4.2 of \cite{LIN} with Brownian perturbation, it is shown $\Psi(x)\leq e^{-Rx}$, where
$R\in(0,1/m)$ solves the equation
$$\lambda\left(\frac{1}{1-mR}-1\right)
=\left(c-\rho\sigma_1\frac{\mu}{\sigma}\right)R
-\frac{\sigma_1^2}{2}(1-\rho^2)R^2
+\frac{\mu^2}{2\sigma^2}.$$
The discounted surplus process of \cite{LIN}
is a jump-diffusion process in a slight different form of ours with $r=0$.
The exponential bound can be obtained
using a constant investment strategy $\pi$ of amount
$\frac{\mu}{R\sigma^2}-\rho\frac{\sigma_1}{\sigma}$ in our model with $r=0$,
and the process $\{e^{-RX^\pi_t}\}_{t\ge 0}$ is a martingale.
\end{remark}

\subsection{The constrained case}
First assume that function $V(x)$ satisfies
 equation \eqref{HJBA} for large $x$. This equation has the form
\begin{equation}
\label{linearIDE}
  \frac{1}{2}(\sigma ^2 A^2+2\rho \sigma \sigma _1 A+\sigma _1^2)V''(x)+ \left[ {c + (\mu  - r)A + rx} \right]V'(x) - M(V)(x) = 0.
\end{equation}
 In the case
of exponential claims, recall $k=1/m$,
 and equation \eqref{linearIDE} can be rewritten as
\begin{equation} \label{explinear}
\begin{split}
&\frac{1}{2}[ \sigma ^2 A^2 +2\rho \sigma \sigma _1 A+\sigma _1^2]V ''(x)+\left[ {c + (\mu  - r)A + rx} \right]V '(x)\\
+&k\lambda\int\limits_0^x {V (x - y)\exp ( - k y)dy} - \lambda V (x) = 0.
\end{split}
\end{equation}
Denote
 $$
g(x): = \int\limits_0^x {V(x - y)\exp ( - ky)dy}.
$$
It is easy to see that
\begin{equation}
\label{g} g'(x) = V (x) - k g(x).
\end{equation}
 If $V$ satisfies \eqref{explinear}, then it satisfies the following equation:
\begin{equation}
\label{eqG}
G'(x) + kG(x) = 0,
\end{equation}
where
\begin{equation} \label{G}
\begin{split}
G(x)=&\frac{1}{2}[ \sigma ^2 A^2 +2\rho \sigma \sigma _1 A+\sigma _1^2]V ''(x)+\left[ {c + (\mu  - r)A + rx} \right]V '(x)\\
+&k\lambda\int\limits_0^x {V (x - y)\exp ( - k y)dy} - \lambda V (x),
\end{split}
\end{equation}
which is the left-hand side of equation
\eqref{explinear}. Then in view of \eqref{g}, equation
\eqref{eqG} can be rewritten as an ordinary differential equation
(ODE) of the 3-rd order:

\begin{equation}\label{ODE}
\begin{split}
0=& V '''(x) + \left[
{\frac{{2c}}{{A_{\rho}^2 }} + \frac{{2(\mu  - r)A}}{{A_{\rho}^2 }}
+ k  + \frac{{2r}}{{A_{\rho}^2 }}x} \right]V ''(x) \\
 &+ \left[ {\frac{{2((r - \lambda ) + k c + k A(\mu  - r))}}{{A_{\rho}^2 }}
 + \frac{{2rk }}{{A_{\rho}^2 }}x} \right]V '(x),
\end{split}
\end{equation}
where
\begin{equation}
\label{Arho}
A_{\rho}^2=\sigma ^2 A^2 +2\rho \sigma \sigma _1 A+\sigma
_1^2.
\end{equation}
Put:
\begin{equation}
\label{a1a2}
a_1  = \frac{{2c}}{{A_{\rho}^2 }} + \frac{{2(\mu  -
r)A}}{{A_{\rho}^2}} + k, \quad  a_2  = \frac{{2r}}{{A_{\rho}^2 }},
\end{equation}
and
\begin{equation}
\label{a3a4}
a_3  = \frac{{2[(r - \lambda ) + k c + k A(\mu  -
r)]}}{{A_{\rho}^2 }}, \quad a_4  = \frac{{2rk }}{{A_{\rho}^2}},
\end{equation}
then the ODE \eqref{ODE} takes the form
\begin{equation}
\label{ODE2}
 \phi '' + (a_2 x + a_1 )\phi ' + (a_4 x + a_3 )\phi  = 0,
\end{equation}
where $ \phi = V'. $ We set
$$ y_1  = \phi, \  y_2  =\phi'.$$
Then
 $$ y'_1  = y_2, \quad y'_2  =  - (a_2 x + a_1 )y_2  -
(a_4 x + a_3 )y_1. $$
Therefore, we obtain the equation in the following matrix form:
\begin{equation}
\label{system}
 y' = (A_1  + A_0 x)y,
\end{equation}
where $ y = (y_1 ,y_2 )^T$, and
 $$ A_0  = \left(
{\begin{array}{*{20}c}
   0 & 0  \\
   { - a_4 } & { - a_2 }  \\
\end{array}} \right),\
A_1  = \left( {\begin{array}{*{20}c}
   0 & 1  \\
   { - a_3 } & { - a_1 }  \\
\end{array}} \right).$$ Rewrite equation \eqref{system} in the form
\begin{equation}
\label{system1}
 x^{ - 1} y' = (A_0 + \frac{{A_1 }}{x})y .
 \end{equation}
This system has an irregular singular point at infinity of the 2-nd
rang (see \cite{Wasov}). Since the matrix $A_0$  has the
eigenvalue zero, then to obtain a principal term of asymptotic
behavior of the solution at infinity, we must find the correction
to the zero eigenvalue by perturbation theory up to $ O(1/x^3 ).$
To do this, we use the method of asymptotic diagonalization for
systems of linear ODE (See \cite{KonNB} and the references therein).
First, we find a diagonalizator of matrix $ A_0$, i.e.
a matrix $D$ such that
$$D^{ - 1} A_0 D = \tilde A_0,$$ where $\tilde
A_0$ is a diagonal matrix. It is easy to show that
$$ D = \left(
{\begin{array}{*{20}c}
   1 & 0  \\
   { - a_4 /a_2 } & 1  \\
\end{array}} \right) = \left( {\begin{array}{*{20}c}
   1 & 0  \\
   { - k } & 1  \\
\end{array}} \right),\
\tilde A_0  = \left( {\begin{array}{*{20}c}
   0 & 0  \\
   0 & { - a_2 }  \\
\end{array}} \right).
$$
Next introduce a change of variables
\begin{equation} \label{z}
 y = D\left( {E + \frac{{N_1
}}{x} + \frac{{N_2 }}{{x^2 }}} \right)z,
\end{equation}
where $z=(z_1,z_2)^T$, $E$ is the $2\times 2$ identity matrix, and
 $ N_1$, $ N_2$ are some $2\times 2$ matrices to be determined
below. Differentiating equation \eqref{z} in $x$, we have
$$ y' =D\left( {E + \frac{{N_1 }}{x} + \frac{{N_2 }}{{x^2 }}} \right)z' -
D\left( {\frac{{N_1 }}{{x^2 }} + \frac{{2N_2 }}{{x^3 }}} \right)z,$$
and we get from \eqref{system1} the equation
\begin{equation}
\label{z'}
x^{ - 1} z' = \left( {E + \frac{{N_1 }}{x} + \frac{{N_2 }}{{x^2
}}} \right)^{ - 1} \left[ {\left( {\tilde A_0  + \frac{{\tilde A_1
}}{x}} \right)\left( {E + \frac{{N_1 }}{x} + \frac{{N_2 }}{{x^2
}}} \right) + O\left( {\frac{1}{{x^3 }}} \right)}
\right]z,
\end{equation}
where
\begin{equation*}
\tilde A_1  = D^{ - 1} A_1 D = \left( {\begin{array}{*{20}c}
   { - k } & 1  \\
   { - a_3  - k \left( {k  - a_1 } \right)} & {k  - a_1 }  \\
\end{array}} \right).
\end{equation*}
Choose matrices $N_1$, $N_2$  in such a way that the equation assumes the form
\begin{equation}
\label{diagonal}
x^{ - 1} z' = \left( {\tilde A_0  + \frac{{\tilde{\tilde{A}}_1 }}{x} + \frac{{\tilde A_2 }}{{x^2 }} + O\left( {\frac{1}{{x^3 }}} \right)} \right)z,
\end{equation}
where  ${\tilde{\tilde{A}}_1 }$ and ${\tilde A_2 }$ are some diagonal matrices.
We let
$$
{\tilde{\tilde{A}}_1 }  = \left( {\begin{array}{*{20}c}
   { - k } & 0  \\
   0 & {k - a_1 }  \\
\end{array}} \right),
$$
(diagonal elements of ${\tilde{\tilde{A}}_1 }$ are the same as those in the matrix $\tilde A_1$) and we determine ${\tilde A_2 }$ below.
Equating the right sides of \eqref{z'} and \eqref{diagonal} we have
$$\left( {\tilde A_0  + \frac{{\tilde A_1 }}{x}} \right)\left( {E + \frac{{N_1 }}{x} + \frac{{N_2 }}{{x^2 }}} \right) + O\left( {\frac{1}{{x^3 }}} \right) = \left( {E + \frac{{N_1 }}{x} + \frac{{N_2 }}{{x^2 }}} \right)\left( {\tilde A_0  + \frac{{\tilde{\tilde{A}}_1 }}{x} + \frac{{\tilde A_2 }}{{x^2 }}} \right).$$
Equating the coefficients of $x^{-1}$ we obtain
$$\tilde A_0 N_1  + \tilde A_1  = {\tilde{\tilde{A}}_1 }  + N_1 \tilde A_0,$$
which yields  $$
N_1  = \left( {\begin{array}{*{20}c}
   0 & { - {1}/{{a_2 }}}  \\
   {{{ - [a_3  + k \left( {k  - a_1 } \right)]}}/{{a_2 }}} & 0  \\
\end{array}} \right).
$$
Equating now the coefficients of $x^{-2}$ we get
$$\tilde A_0 N_2  + \tilde A_1 N_1  = \tilde A_2  + N_1  {\tilde{\tilde{A}}_1 } + N_2 \tilde A_0,$$
and \begin{equation*}
\begin{split}
&N_2  = \left( {\begin{array}{*{20}c}
   0 & { - \left( {2k - a_1 } \right)/{{a_2 ^2 }}}  \\
   {{{[a_3+k(k-a_1) ]}}\left( {a_1-2k } \right)/{{a_2^2 }}} & 0  \\
\end{array}} \right),\\
&\tilde A_2  = \left( {\begin{array}{*{20}c}
   {{\lambda }/{r} - 1} & 0  \\
   0 & {1 - {\lambda }/{r}}  \\
\end{array}} \right).
\end{split}
\end{equation*}
The system \eqref{diagonal} is  asymptotically equivalent to the following system (see \cite{Bellman}):
\begin{equation}
\label{diagonal2}
\tilde z' = x\left( {\tilde A_0  + \frac{{\tilde{\tilde{A}}_1 }}{x} + \frac{{\tilde A_2 }}{{x^2 }} } \right)\tilde z,
\end{equation}
where $\tilde z=(\tilde z_1, \tilde z_2)^T$, which is separated into two independent equations:
\begin{equation*}
\begin{split}
&\tilde z_1 ^\prime   = \left( { - k  + \frac{{{\lambda }/{r} - 1}}{x}} \right)\tilde z_1, \\
&\tilde z_2 ^\prime   = \left( { - \frac{{2r}}{{A_{\rho}^2 }}x -
{\frac{{2c}}{{A_{\rho}^2 }} - \frac{{2(\mu  - r)A}}{{A_{\rho}^2
}}}  + \frac{{1 - {\lambda }/{r}}}{x}} \right)\tilde z_2.
\end{split}
\end{equation*}
For $x\rightarrow\infty$,  the solutions of these equations have the following form:
\begin{equation*}
\begin{split}
&\tilde z_1  = C_1 e^{ - k \,x} x^{{\lambda }/{r} - 1} (1 + o(1)),\\
&\tilde z_2  = C_2 {\mathop{\rm e}\nolimits} ^{ -
{\frac{r}{{A_{\rho}^2 }}x^2  - \frac{{2(c + (\mu  -
r)A)}}{{A_{\rho}^2 }}x}} x^{1 - {\lambda }/{r}} (1 + o(1)),
\end{split}
\end{equation*}
for some constants $C_1$ and $C_2$. The same representations are true for the solution of \eqref{diagonal}. Notice
\begin{equation*}
\begin{split}
&y_1  = z_1  + \left( { - \frac{{A_{\rho}^2 }}{{2rx}} + \frac{{l_2 }}{{x^2 }}} \right)z_2, \\
&y_2  = \left( { - k  - \frac{{r - \lambda }}{{rx}} + \frac{{l_3 }}{{x^2 }}} \right)z_1  + \left[ {1 + k\left( {\frac{{A_{\rho}^2}}{{2rx}} - \frac{{l_2 }}{{x^2 }}}\right)} \right]z_2,
\end{split}
\end{equation*}
where $l_2$, $l_3$ are the elements of matrix $N_2$:
$$l_2=- \frac{1}{{a_2 ^2 }}\left( {2k - a_1 } \right), \ l_3=\frac{{r-\lambda }}{{ra_2 }}\left( {a_1-2k } \right).$$
 Therefore, considering the above notation for nondecreasing function $V$ satisfying (\ref{HJBA}) for large $x$,
we conclude
\begin{equation}
\label{V'infty} V'(x) = C_1 e^{ - x/m} x^{{\lambda }/{r} - 1} (1 +
o(1)), \quad x\rightarrow\infty,
\end{equation}
and
$$V''(x)= \left( { - \frac{1}{m}  - \frac{{r - \lambda }}{{rx}} + \frac{{l_3 }}{{x^2 }}} \right)C_1 e^{ - x/m} x^{{\lambda }/{r} - 1} (1 + o(1)),\ \  x\rightarrow\infty,$$
where $C_1>0$.
Thus $$a_V(x)= \frac{(\mu-r)m}{\sigma ^2}-\rho\frac{\sigma_1}{\sigma}-\left(1-\frac{\lambda}{r}\right)\frac{(\mu-r)m^2}{\sigma ^2}\frac{1}{x}(1+o(1)),\ \  x\rightarrow\infty.$$

Note that the same conclusions are true for the case $A=0$,
 if we consider the corresponding values of $a_i$, $i=1,...,4$, from \eqref{a1a2}, \eqref{a3a4} and $A_{\rho}$ from \eqref{Arho} (in particular,  $A_{\rho}=\sigma _1$ in this case).

Write
$$\rho_3=\frac{m(\mu-r)}{\sigma\sigma_1},\ \rho_4=\rho_3-\frac{A\sigma}{\sigma_1}.$$

If $A<\frac{(\mu-r)m}{\sigma ^2}-\rho\frac{\sigma_1}{\sigma}$,
or $\rho<\rho_4$, we have $V''(x)<0$ and $a_V(x)> A$ for large $x$, thus
$a^*_V(x)=A$ and $V$ solves \eqref{HJBA}.

If $\rho_4<\rho<\rho_3$ (i.e., $0<\frac{(\mu-r)m}{\sigma ^2}-\rho\frac{\sigma_1}{\sigma}<A$),
 from the asymptotic representation \eqref{asymp_a*} for optimal strategy at
large values of the surplus in the unconstrained case,
we see $V''(x)<0$ and $0<a_V(x)<A$ for large values of the surplus;
then the optimizer is $a^*_V(x)=a_V(x)$, where
$V$ solves \eqref{HJBav}.

If $\rho>\rho_3$ ($\frac{(\mu-r)m}{\sigma ^2}-\rho\frac{\sigma_1}{\sigma}<0$),
we have $V''(x)<0$ and $a_V(x)<0$ for large $x$, thus
$a^*_V(x)=0$ and $V$ solves \eqref{HJB0}.

If $\rho=\rho_4$ (i.e., $A=\frac{(\mu-r)m}{\sigma ^2}-\rho\frac{\sigma_1}{\sigma}$), then $V''(x)<0$, $a_V(x)\geq A$ and  $a^*_V(x)=A$ for large $x$ if $(\lambda-r)$ is positive, or  $0<a_V(x)<A$ and $a^*_V(x)=a_V(x)$ if $(\lambda -r)$  is negative. Then $V$ solves \eqref{HJBA} or \eqref{HJBav} respectively. If $(\lambda-r)=0$, one of these cases takes place depending on others parameters and we omit further discussions.

If $\rho=\rho_3$ ($\frac{(\mu-r)m}{\sigma ^2}-\rho\frac{\sigma_1}{\sigma}=0$),
we have $V''(x)<0$,  $0<a_V(x)<A$ and  $a^*_V(x)=a_V(x)$ for large $x$ if $(\lambda -r)$ is positive, or  $a_V(x)\geq 0$ and $a^*_V(x)=0$ if
$(\lambda -r)$  is negative. Then $V$ solves \eqref{HJBav} or \eqref{HJB0} respectively. If $(\lambda-r)=0$, one of these cases takes place depending on others parameters.

We have the following theorems on the optimal strategy and the value function:
\begin{theorem}\label{th53}
For large $x$, it holds
\begin{eqnarray}\label{astar-infty}
a^*_V(x)=\begin{cases} \left[\frac{(\mu-r)m}{\sigma
^2}-\rho\frac{\sigma_1}{\sigma}\right](1+O(\frac{1}{x})),
\;  &\ A\geq \frac{(\mu-r)m}{\sigma ^2}-\rho\frac{\sigma_1}{\sigma}\geq 0; \\
\qquad \qquad 0, &\ \frac{(\mu-r)m}{\sigma ^2}-\rho\frac{\sigma_1}{\sigma}<0;\\
\qquad \qquad A, &\ A<\frac{(\mu-r)m}{\sigma ^2}-\rho\frac{\sigma_1}{\sigma}.
\end{cases}
\end{eqnarray}
Moreover, if $A > \frac{(\mu-r)m}{\sigma ^2}-\rho\frac{\sigma_1}{\sigma} > 0$ in the first case, more exact relation \eqref{asymp_a*} is fulfilled.
\end{theorem}

\begin{theorem}\label{th54}
 It holds $$
\delta (x) = 1-K_2 e^{ - x/m} x^{{\lambda }/{r} - 1} (1 + o(1)), \
x\rightarrow\infty,$$
for some constant $K_2=K_2(A)>0$.
\end{theorem}
\begin{remark}\label{rmk-54}
From the analysis in this section,
we see that under any investment strategy with a fixed amount invested in
the risky asset, the survival probability function
has a limiting expression with the same principal term (product
of exponential and power functions) as in Theorem~\ref{th54}.
We note that the results in Theorems~\ref{th53} and \ref{th54}
remain valid in the model without perturbation ($\sigma_1=0$).
We also note that in the case without risky investment and perturbation
 ($a_s\equiv 0$ and $\sigma_1=0$), the ruin probability function has the following form (See, e.g. \cite{PG} and \cite{SG}):
$$\Psi(x)=\frac{\int_x^\infty e^{-u/m}\big(1+ru/c\big)^{\lambda/r-1}du}
{c/\lambda+\int_0^\infty e^{-u/m}\big(1+ru/c\big)^{\lambda/r-1}du},$$
which implies $\Psi(x) = K e^{ - x/m} x^{{\lambda}/{r}-1}(1 + o(1))$, $x\rightarrow\infty$, for some $K>0$.
\end{remark}

\section{Numerical Examples and Conclusions}
In this section, we give two numerical examples and a few concluding remarks.
In the examples, we consider the case of unconstrained investment.
Computations using the asymptotic results and the operator \eqref{vLv}
are conducted for various claim distributions.
\begin{example}In this example, the parameters are given by the following:
$\mu=0.42$,
$r=0.32$,
$c=0.36$,
$\lambda=0.3$,
$\rho=-0.2$,
$\sigma=0.1$,
$\sigma_1=0.2$, and $k=1$.
We give calculations for cases with
exponential, half-normal and log-normal claim distributions.
For all cases of different distributions, we have
 the following asymptotic
result for low surplus levels
$$a_V^*(x)\approx 0.8542115-0.02039470x(1+o(1)),\ \ x\rightarrow 0,$$
using \eqref{astar0} and Theorem~\ref{thm42}.
The exponential claim distribution has mean $1$ with tail probability function
 $H(x)=e^{-x}$.
The half-normal claim
distribution has density and tail-probability functions given below:
$$f(x)=\frac{1}{v\sqrt{\pi/2}}e^{-\frac{x^2}{2v^2}},\ \
H(x)=2\left[1-\Phi\left(\frac{x}{v}\right)\right],\ \ x>0,$$
where $\Phi$ is the standard normal distribution function and $v$ is a parameter. We set $v=\sqrt{\pi/2}$ and then
 the mean of the distribution is $v\sqrt{2/\pi}=1$.
The log-normal distribution has density and tail probability functions
$$f(x)=\frac{1}{\sqrt{2\pi}vx}e^{-\frac{(lnx-u)^2}{2v^2}},\ \ H(x)=1-\Phi\left(\frac{\log(x)-u}{v}\right), \ \ x>0,$$
with parameters $v>0$ and $u\in(-\infty,\infty)$. We set $v=1$ and $u=-0.5$
and hence the mean is $e^{u+v^2/2}=1$.
For these claim distributions,
the optimal investment controls calculated using \eqref{aw} and the operator
\eqref{vLv} numerically are given in Figures~\ref{eg1-0} and ~\ref{eg1-infty}.
With the given exponential claim distribution,
it holds the following asymptotic result for large surplus levels:
 $$a_V^*(x)\approx 10.4-0.625\frac{1}{x}(1+o(1)),\ \ x\rightarrow \infty,$$
using \eqref{asymp_a*}.
\end{example}
\begin{example}In this example, the parameters are given by the following:
$\mu=0.2$,
$r=0.12$,
$c=0.5$,
$\lambda=0.3$,
$\rho=0.15$,
$\sigma=0.9$,
$\sigma_1=0.5$, and $k=0.5$.
We give calculations for cases with
exponential, Weibull and Pareto claim distributions.
For all cases of different distributions, we have the following asymptotic
result for low surplus levels
$$a_V^*(x)\approx -0.05274736+0.01112835x(1+o(1)),\ \ x\rightarrow 0,$$
using \eqref{astar0} and Theorem~\ref{thm42}.
The exponential distribution has mean $2$ and $H(x)=e^{-0.5x}$.
The Weibull distribution has density and tail-probability functions
$$f(x)=\frac{v}{u}\left(\frac{x}{u}\right)^{v-1}e^{-\left(\frac{x}{u}\right)^v},
\ \ H(x)=e^{-\left(\frac{x}{u}\right)^v},\ \ x>0,$$
where $u$ and $v$ are parameters. We set $u=1$, $v=1/2$. So the mean
of the distribution is $u\Gamma(1+1/v)=2$. The Pareto claim distribution
has density and tail probability functions
$$f(x)=\frac{vu^v}{(u+x)^{v+1}}, \ \ H(x)=\left(\frac{u}{u+x}\right)^v,\ \ x>0,$$
where $u$ and $v$ are parameters.
We set $u=2$, $v=2$; so the mean is $u/(v-1)=2$.
For these claim distributions,
the optimal investment controls calculated using \eqref{aw} and the operator
\eqref{vLv} numerically are given in Figures~\ref{eg2-0} and ~\ref{eg2-infty}
(We note that the optimal investment strategies in this example barely show a
 difference in Figure~\ref{eg2-0} in the two cases with the exponential and Pareto claim distributions).
In the case with the exponential claim distribution,
we have the following asymptotic result for large surplus levels
 $$a_V^*(x)\approx 0.163580+0.740741\frac{1}{x}(1+o(1)),
 \ \ x\rightarrow\infty,$$
using \eqref{asymp_a*}.
\end{example}

\begin{figure}[ht]
\centering
\begin{minipage}[b]{0.8\linewidth}
\includegraphics[width=1\textwidth]{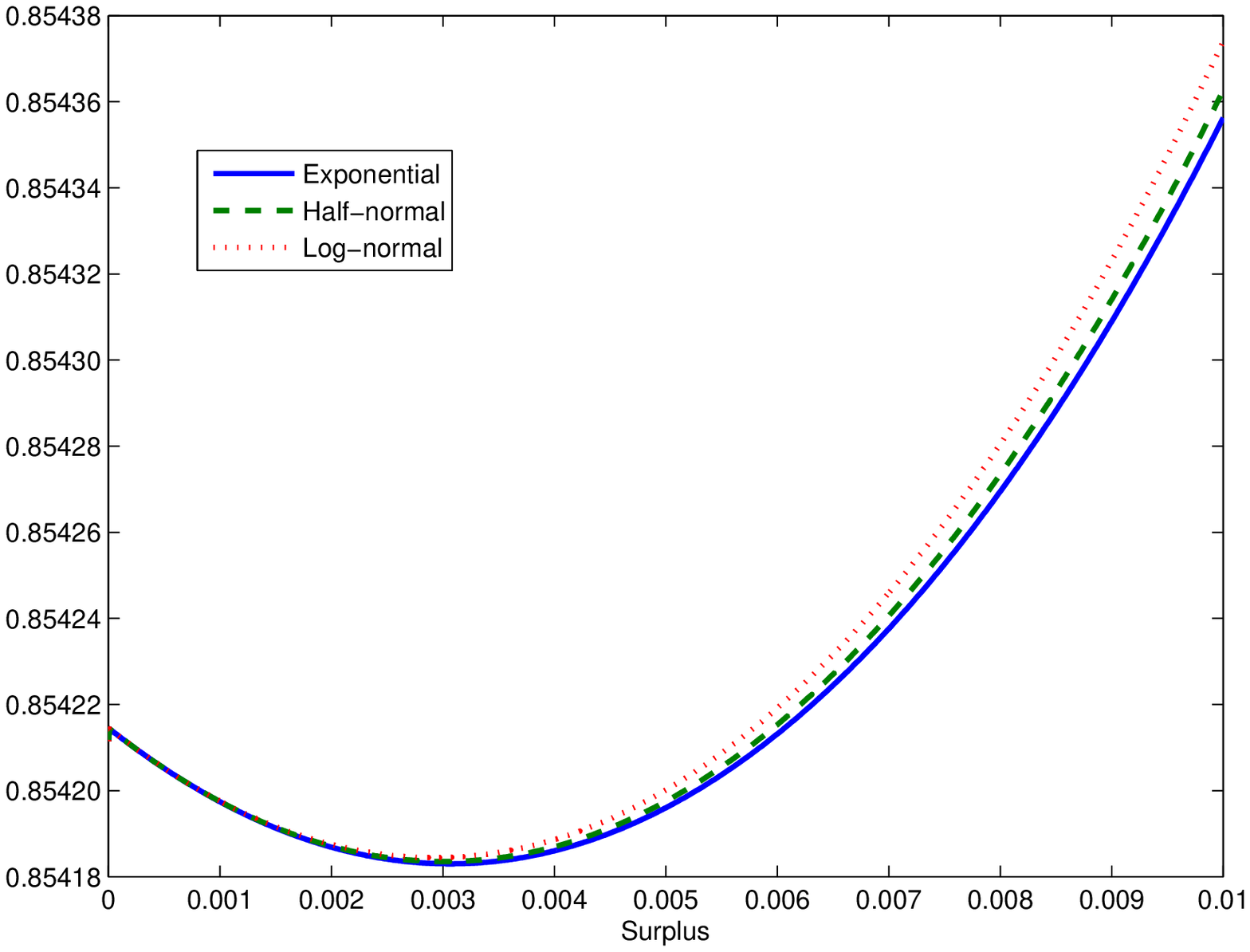}
                \caption{Optimal investment at low surplus - Example 1}
                \label{eg1-0}
\end{minipage}
\begin{minipage}[b]{0.8\linewidth}
\includegraphics[width=1\textwidth]{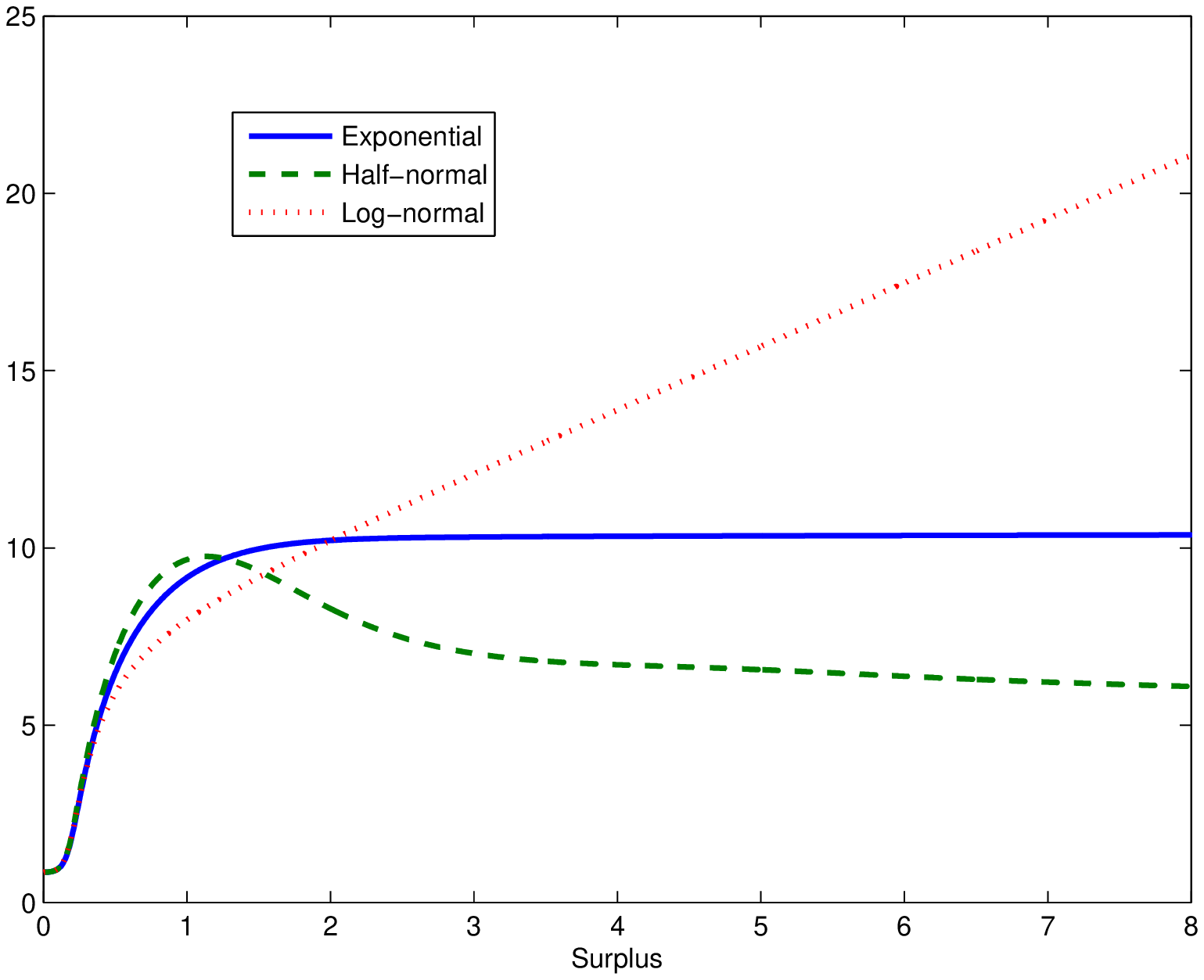}
               \caption{Optimal investment at large surplus - Example 1}
                \label{eg1-infty}
\end{minipage}
\end{figure}

\begin{figure}[ht]
\begin{minipage}[b]{0.8\linewidth}
 \includegraphics[width=1\textwidth]{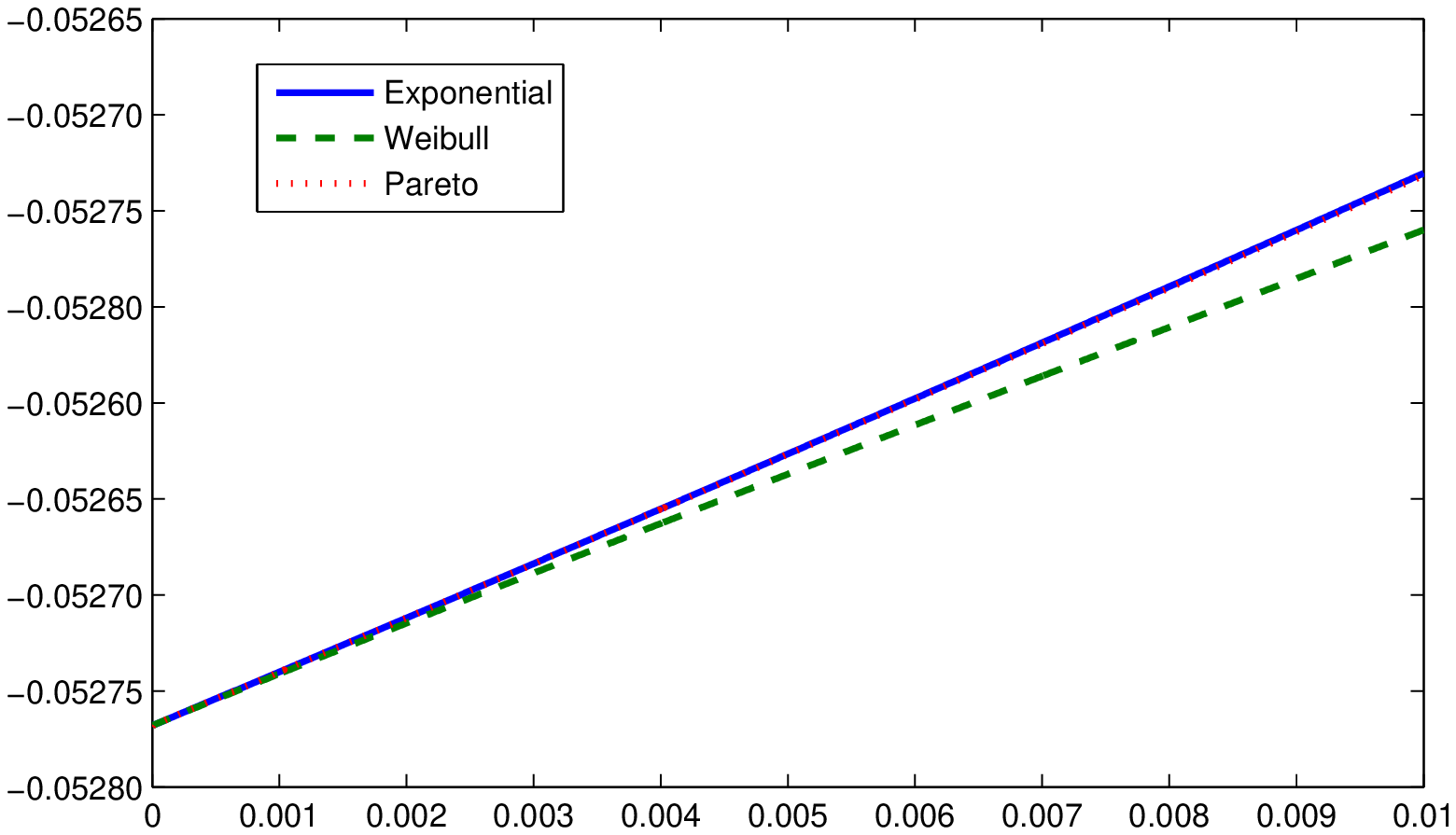}
                \caption{Optimal investment at low surplus - Example 2}
                \label{eg2-0}
\end{minipage}
\begin{minipage}[b]{0.8\linewidth}
 \includegraphics[width=1\textwidth]{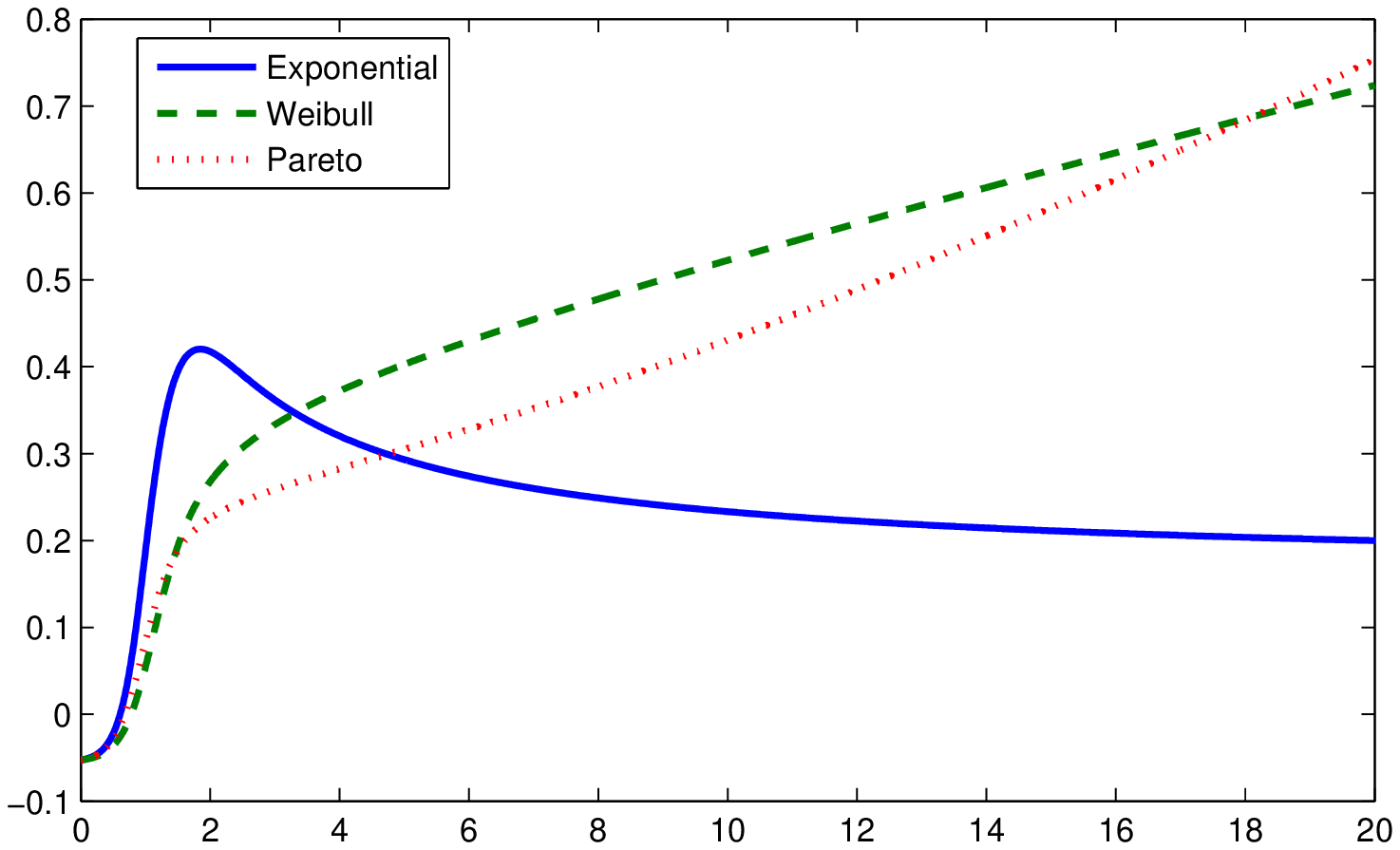}
               \caption{Optimal investment at large surplus - Example 2}
               \label{eg2-infty}
\end{minipage}
\end{figure}

In this paper, we study the
optimal investment control problem
under the scenario of ruin minimization.
The surplus is modeled by a perturbed Cramer-Lundberg process.
Investment control with a Black-Scholes stock
and a risk-free asset is considered.
We prove the existence of a classical solution to the HJB equation
in both cases of investment using operators.
In the constrained investment case, for low surplus levels,
we find parameter conditions under which
the optimal investment amount takes values
 $0$ (no investment in the risky asset)
or $A$ (maximal level of risky investment), or it
tends to a fixed level.
In the unconstrained case, we show that
 the optimal investment amount approaches to a fixed level
at a rate of order $x$ as the surplus level $x$ goes to $0$.
We also show that the maximal survival probability tends
to $0$ at a rate of order $x$.
In the case with exponential claims,
 we give new asymptotic results for large surplus values.
We prove that the optimal investment amount tends to a fixed level
at a rate of $1/x$ as the surplus level $x$ tends to infinity.
We also prove that the minimal ruin probability function
has a limit expression of $e^{ - x/m} x^{{\lambda }/{r} - 1}$.

In general, the optimal investment control and the maximal survival probability
function are not analytically tractable under the jump-diffusion model, i.e.,
it is usually unable to give explicit expressions for them.
Thus the asymptotic results in this paper
provide convenient and insightful calculations
when finding the optimal investment control
and the maximal survival probability.

\vspace{0.5 cm}
\centerline{\textbf{Acknowledgements} }
The first author of this research
was supported by the Russian Fund of Basic Research,
grants RFBR 13-01-00784 and RFBR 11-01-00219, and  the
 International Laboratory of Quantitative Finance, NRU HSE, RF government grant, ag. 14.A12.31.0007 (TB).
This paper was also written partially during a visit (TB) to the
Hausdorff Research Institute for Mathematics at the University of
Bonn in the framework of the Trimester Program ``Stochastic
Dynamics in Economics and Finance" (TB). The second author (SL) acknowledges the
support of a Professional Development Assignment
and a Summer Research Fellowship
from the University of Northern Iowa.


{\bf Appendix}

We give a sketched proof for the asymptotic series representation
of $\widetilde{a}_V$ in \eqref{series}.

Note that the function $\tilde{a}_V(x)$ is a solution to equation \eqref{eq_tilde_a}, which is
a nonlinear ODE copied below (in terms of $\phi$):
\begin{equation}
 \label{eq_phi}
\begin{split}
[\sigma ^2 \phi^2(x)+\sigma _\rho^2]\phi '(x)
=&-\frac{\sigma ^2}{m} \phi^3(x)-2\left[r-\lambda+ \frac{c_\rho}{m} -\frac{{(\mu-r)^2 }}{2\sigma ^2} + \frac{r}{m}x \right]
\frac{\sigma ^2}{\mu-r}\phi^2(x)\\
&+2\left(c_{\rho} + rx + \frac{1}{2m}\sigma _\rho^2\right)\phi(x)
  - \sigma_\rho^2 \frac{{\mu  - r}}{{\sigma ^2 }}.
\end{split}
 \end{equation}

We see that the equation \eqref{eq_phi} is asymptotically autonomous by a change of variables $y = x^2/2$ and letting $y\to \infty$.
 This autonomous equation has two finite stationary points. One of these is a stable point equal to $\tilde a_0$ and the other is an unstable point.
As a result, a solution of \eqref{eq_phi} must have a finite limit
equal to the stable point or tend to
 (+ or -) infinity as $x\rightarrow \infty$
 (See \cite{Bellman}, \cite{LY} and \cite{Wasov}).

We characterize at first one finite limit solution to \eqref{eq_phi}
given by a series
 \begin{equation}
\label{series1} w(x)=\Sigma _{k=0}^{\infty}\widetilde{a}_k/x^k,
\end{equation}
where the coefficients $\widetilde  {a}_k$ are given by (using
\eqref{eq_phi}):
 \begin{equation}\label{coeff}
 \widetilde{a}_0={(\mu-r)m}/{\sigma ^2}>0, \quad  \;\; \
\widetilde{a}_1=-\left(1-{\lambda}/{r}\right){(\mu-r)m^2}/{\sigma
^2},... .
\end{equation}
Suppose $\widetilde{w}(x)$ is another finite-limit-solution to \eqref{eq_phi}.
Define $b(x)=\widetilde{w}(x)-w(x)$.
The function  $b(x)$ solves the ODE
\begin{equation}\begin{split}\label{eqb}
&\sigma ^2(b^2+2wb)w^\prime+[\sigma^2(b^2+2wb+w^2)+\sigma _\rho^2]b^\prime\\
 =&-\frac{\sigma^2}{m}\left(b^3+3b^2w+3bw^2\right)-2(A+Bx)(b^2+2bw)+2(C+rx)b,
\end{split}
\end{equation}
where
$$A=\left [r-\lambda +\frac{c_\rho}{m}-\frac{{(\mu-r)^2 }}{2\sigma
^2}\right ] \frac{\sigma ^2}{\mu-r},\quad B= \frac{r\sigma
^2}{m(\mu-r)}, \quad C=c_{\rho}  + \frac{\sigma _\rho^2}{2m}. $$
Further let us linearize the ODE \eqref{eqb} on
$b(x)$ with $b(x)\to 0$, $x\to\infty$.  Taking
into account the principal linear terms of the expansion in powers
of $1/x$, we obtain
\begin{equation}\label{beq}
\widetilde{b}^\prime=x\left (d_0+\frac{d_1}{x}+\frac{d_2}{x^2}\right)\widetilde{b}, \quad
x\gg 1.
\end{equation}
Here $d_0=-2r/[\sigma _\rho^2 + (\mu-r)^2 m^2/\sigma ^2]<0$, and
we omit  expressions for the coefficients $d_1$ and $d_2$.
A general solution of the ODE \eqref{beq}
has the form $\widetilde{b}(x,D)=Dx^{d_2}\exp{(d_0 x^2/2
+ d_1 x)}$, where $D$ is an arbitrary constant. Then the
nonlinear ODE \eqref{eqb} has a one-parameter family of
solutions which can be represented by the parametric Lyapunov
series in terms of the integer powers of $\widetilde{b}(x,D)$.
Thus any finite-limit solution to \eqref{eq_phi}
has the following asymptotic representation:
\begin{equation}\label{aD}
\widetilde{w}(x,D)\thicksim\Sigma
_{k=0}^{\infty}\widetilde{a}_kx^{-k} + Dx^{d_2}\exp{(d_0x^2/2+d_1x)} (1+o(1)),
\,\, x\to\infty,
\end{equation}
where $d_0<0$ (see above), $\widetilde{a}_0>0$ and
$\widetilde{a}_1$ are defined in \eqref{coeff}.

Next we prove by contradiction
that $\tilde{a}_V$ must be a finite-limit solution of \eqref{eq_phi}.
Notice that $\tilde{a}_V>0$ when $\mu-r>0$ and $\tilde{a}_V<0$ when $\mu-r<0$.
For the case $\mu-r>0$, if we assume
$\underset{x\rightarrow\infty}{\lim}\tilde{a}_V(x)=\infty$,
dividing both sides of \eqref{eq_tilde_a} by
$[\sigma ^2 \tilde{a}_V^2(x)+\sigma _\rho^2]$
and letting $x\rightarrow\infty$, we then have
 $$\underset{x\rightarrow\infty}{\lim}\tilde{a}_V'(x)
=\underset{x\rightarrow\infty}{\lim}\left\{
-\frac{1}{m}\tilde{a}_V-2\left[r-\lambda+ \frac{c_\rho}{m} -\frac{{(\mu-r)^2 }}{2\sigma ^2} \right]
\frac{1}{\mu-r}
-2x\frac{\frac{r\sigma^2}{m(\mu-r)}\tilde{a}_V^2-r\tilde{a}_V}{\sigma ^2 \tilde{a}_V^2(x)+\sigma _\rho^2}\right\}
=-\infty,$$
wherefrom it holds $\underset{x\rightarrow\infty}{\lim}\tilde{a}_V(x)=-\infty$,
which contradicts to the assumption $\underset{x\rightarrow\infty}{\lim}\tilde{a}_V(x)=\infty$!
For the case $\mu-r<0$, if we assume 
$\underset{x\rightarrow\infty}{\lim}\tilde{a}_V(x)=-\infty$,
then similarly from \eqref{eq_tilde_a} we have 
 $$\underset{x\rightarrow\infty}{\lim}\tilde{a}_V'(x)
=\infty,$$
which implies  $\underset{x\rightarrow\infty}{\lim}\tilde{a}_V(x)=\infty$, 
leading to contradiction! 
So we conclude that
 $\tilde{a}_V$ is a finite-limit solution of \eqref{eq_phi} which has
the form of \eqref{aD}. We then obtain \eqref{series}.
\end{document}